\documentclass[journal]{IEEEtran}

\usepackage[english]{babel}
\usepackage{booktabs}

\usepackage{ifpdf}

\usepackage{cite} 
\usepackage{url}
\usepackage{hyperref}

\ifCLASSINFOpdf
	\usepackage[pdftex]{graphicx}
	\graphicspath{{./figures/}}
\else
	\usepackage[dvips]{graphicx}
	\graphicspath{./figures/}
\fi
\usepackage{color}

\usepackage{pgf, tikz, pgfplots}
\usetikzlibrary{shapes, arrows, automata}
\usetikzlibrary{calc,hobby,decorations}

\usepackage[cmex10]{amsmath}
\usepackage{amsfonts, amssymb, amsthm}
\usepackage{mathrsfs}

\usepackage{algorithm} 
\usepackage{algorithmic}
\usepackage{enumerate}
\usepackage{multirow}
\usepackage{rotating}
\usepackage{subcaption}
\usepackage{needspace}

\input{mySymbol.sty}

\definecolor{penndarkestblue}{cmyk}{1,0.74,0,0.77}
\definecolor{penndarkerblue}{cmyk}{1,0.74,0,0.70}
\definecolor{pennblue}{cmyk}{0.99,0.66,0,0.57} 
\definecolor{pennlighterblue}{cmyk}{0.98,0.44,0,0.35}
\definecolor{pennlightestblue}{cmyk}{0.38,0.17,0,0.17} 

\definecolor{penndarkestred}{cmyk}{0,1,0.89,0.66}
\definecolor{penndarkerred}{cmyk}{0,1,0.88,0.55}
\definecolor{pennred}{cmyk}{0,1,0.83,0.42} 
\definecolor{pennlighterred}{cmyk}{0,1,0.6,0.24}
\definecolor{pennlightestred}{cmyk}{0,0.43,0.26,0.12} 

\definecolor{penndarkestgreen}{cmyk}{1,0,1,0.68}
\definecolor{penndarkergreen}{cmyk}{1,0,1,0.57}
\definecolor{penngreen}{cmyk}{1,0,1,0.44} 
\definecolor{pennlightergreen}{cmyk}{1,0,1,0.25}
\definecolor{pennlightestgreen}{cmyk}{0.43,0,0.43,0.13}

\definecolor{penndarkestorange}{cmyk}{0,0.65,1,0.49}
\definecolor{penndarkerorange}{cmyk}{0,0.65,1,0.33}
\definecolor{pennorange}{cmyk}{0,0.54,1,0.24} 
\definecolor{pennlighterorange}{cmyk}{0,0.32,1,0.13}
\definecolor{pennlightestorange}{cmyk}{0,0.15,0.46,0.06}
	
\definecolor{penndarkestpurple}{cmyk}{0,1,0.11,0.86}
\definecolor{penndarkerpurple}{cmyk}{0,1,0.13,0.82}
\definecolor{pennpurple}{cmyk}{0,1,0.11,0.71} 
\definecolor{pennlighterpurple}{cmyk}{0,1,0.05,0.46}
\definecolor{pennlightestpurple}{cmyk}{0,0.35,0.02,0.23}
	
\definecolor{pennyellow}{cmyk}{0,0.20,1,0.05} 
\definecolor{pennlightgray1}{cmyk}{0,0,0,0.05}
\definecolor{pennlightgray3}{cmyk}{0.01,0.01,0,0.18}
\definecolor{pennmediumgray1}{cmyk}{0.04,0.03,0,0.31}
\definecolor{pennmediumgray4}{cmyk}{0.08,0.06,0,0.54}
\definecolor{penndarkgray2}{cmyk}{0.09,0.07,0,0.71}
\definecolor{penndarkgray4}{cmyk}{0.1,0.1,0,0.92}

\def\SO3{\mathrm{SO(3)}}

\newtheorem{assumption}{\hspace{0pt}\bf Assumption \hspace{-0.15cm}}
\newtheorem{lemma}{\hspace{0pt}\bf Lemma}
\newtheorem{proposition}{\hspace{0pt}\bf Proposition}

\newtheorem{theorem}{\hspace{0pt}\bf Theorem}
\newtheorem{corollary}{\hspace{0pt}\bf Corollary}

\newtheorem{remark}{\hspace{0pt}\bf Remark}

\newtheorem{definition}{\hspace{0pt}\bf Definition}

\begin{document}

\title{Learning Stochastic Graph Neural Networks \\ with Constrained Variance}

\author{\IEEEauthorblockN{Zhan Gao$^{ \dagger}$ and Elvin Isufi$^{\ddagger }$}\\
\thanks{Preliminary results was presented in ICASSP 2021 \cite{gao2021variance}. $^{\dagger}$Department of Computer Science Technology, University of Cambridge, Cambridge, UK (Email: zg292@cam.ac.uk). $^{\ddagger }$Department of Intelligent Systems, Delft University of Technology, Delft, The Netherlands (Email: e.isufi-1@tudelft.nl). The work of E. Isufi is supported by the TU Delft AI Labs programme.}}

\markboth{IEEE TRANSACTIONS ON SIGNAL PROCESSING (SUBMITTED)}%
{Learning Stochastic Graph Neural Networks with Constrained Variance}

\maketitle

\begin{abstract}
Stochastic graph neural networks (SGNNs) are information processing architectures that learn representations from data over random graphs. SGNNs are trained with respect to the expected performance, which comes with no guarantee about deviations of particular output realizations around the optimal expectation. To overcome this issue, we propose a variance-constrained optimization problem for SGNNs, balancing the expected performance and the stochastic deviation. An alternating primal-dual learning procedure is undertaken that solves the problem by updating the SGNN parameters with gradient descent and the dual variable with gradient ascent. To characterize the explicit effect of the variance-constrained learning, we analyze theoretically the variance of the SGNN output and identify a trade-off between the stochastic robustness and the discrimination power. We further analyze the duality gap of the variance-constrained optimization problem and the converging behavior of the primal-dual learning procedure. The former indicates the optimality loss induced by the dual transformation and the latter characterizes the limiting error of the iterative algorithm, both of which guarantee the performance of the variance-constrained learning. Through numerical simulations, we corroborate our theoretical findings and observe a strong expected performance with a controllable variance.
\end{abstract}


\IEEEpeerreviewmaketitle

\section{Introduction} \label{sec:intro}


Networked data exhibits an irregular structure inherent in its underlying topology and can be represented as signals residing on the nodes of a graph \cite{ortega2018graph}. Graph neural networks (GNNs) exploit this structural information to model task-relevant representations from graph signals \cite{Scarselli2009, Defferrard2016, gama2020graphs, Wu2019}, which have found applications in recommender systems \cite{Ying2018, FanAGraph}, multi-agent coordination \cite{tolstaya2020learning, gao2021wide}, wireless communications \cite{eisen2020optimal, gao2020resource}, etc. The success of GNNs can be attributed to their ability of leveraging the coupling between the signal and the graph, but the latter may be prone to perturbations such as adversarial attacks, link losses in distributed communications, or topological estimation errors. In these settings, the graph encountered during testing differs from the one used during training; hence, questioning the stability to such perturbations.

The stability of GNNs to graph perturbations has been investigated in \cite{Gama20-Stability, kenlay2021interpretable, kenlay2021stability, levie2021transferability, parada2021algebraic}. The work in \cite{Gama20-Stability} characterized the stability of GNNs to absolute and relative perturbations, which shows GNNs can be both stable to small perturbations and discriminative at high graph frequencies. Authors in \cite{kenlay2021interpretable, kenlay2021stability} analyzed the stability of graph filters --the linear inner working mechanism of GNNs that captures the graph-data coupling-- and GNNs under structural perturbations and provided interpretable stability bounds. The work in \cite{levie2021transferability} established GNNs can extract similar representations on graphs that describe the same phenomenon, while authors in \cite{parada2021algebraic} extended the stability results to the algebraic neural network where GNNs can be seen as a particular case.

The above works discuss the GNN stability w.r.t. small deterministic perturbations. However, the graph can often change randomly, resulting in stochastic perturbations that cannot be addressed with the above analysis.  Stochastic perturbations appear when GNNs are implemented distributively on physical networks \cite{Isufi17, Zou2013, Shuman2018}, where communication links fall with a certain probability due to channel fading effects, leading to random communication graphs \cite{kar2008sensor, antonelli2014decentralized, deng2016latent}. Other cases, in which GNNs operate on stochastic graphs, involve recommender systems, where the graph stochasticity is introduced to improve the recommendation diversity \cite{monti2017geometric, berg2017graph, isufi2021accuracy}. The impact of stochastic perturbations on graph filters 
has been analyzed in \cite{Isufi17}, while \cite{saad2020quantization} extended the analysis to scenarios with both graph randomness and quantization effects. The work in \cite{nguyen2021stability} studied the stability of low pass graph filters to edge rewiring on the stochastic block model. Authors in \cite{gao2021stability} characterized the stability of GNNs to stochastic perturbations and identified the role played by the filter, nonlinearity, and architecture. 

To alleviate the performance degradation induced by stochastic perturbations, the work in \cite{gao2021stochastic} proposed stochastic graph neural networks (SGNNs) that account for the graph stochasticity during training. Learning with uncertainty makes the trained model robust to perturbations encountered during testing, and thus endows the SGNN with robust transference properties. The graph stochasticity has also been considered during training as a regularization technique to prevent over-smoothing \cite{feng2020graph} or as a data augmentation technique to avoid over-fitting \cite{rong2019dropedge, gao2021training}.

While improving the stability to perturbations, training an SGNN implies optimizing the expected performance w.r.t. the random topology in an empirical risk minimization framework \cite{gao2021stochastic}. However, such a strategy does not provide any guarantee about the deviation of a single SGNN realization around the optimal expectation; hence, an undesirable performance may appear in individual realizations, even when the expected performance is satisfactory. To control such a deviation, we propose a variance-constrained learning strategy for SGNNs that optimizes the expected performance while constraining the output variance. 
The proposed strategy adheres to solving a stochastic optimization problem subject to a variance constraint. This is a challenging problem because of the constraint, the stochastic nature of the topology, and the non-convexity of the SGNN. Following recent advances in constrained learning \cite{chamon2020functional}, we adopt a primal-dual learning procedure to solve the problem. 
To study the effect of such strategy on the SGNN learning capacity, we characterize its output variance theoretically to identify the implicit trade-off between the improved deviation robustness and the degraded discrimination power. Our detailed contribution is threefold:

\smallskip
\begin{enumerate}[(i)]
	
	\item \emph{Variance-constrained learning (Section \ref{sec_problem})}: We formulate a constrained stochastic optimization problem that balances the expected performance with the stochastic deviation. We propose a primal-dual learning procedure to solve the problem, which updates alternatively the primal SGNN parameters with gradient descent and the dual variable with gradient ascent to search for a saddle point. We show this strategy acts as a self-learning variance regularizer.
	
	\item \emph{Variance and discrimination (Section \ref{DOL})}: We analyze theoretically the variance of the SGNN output and identify the effect of the filter property, graph stochasticity and architecture. The variance-constrained learning restricts the variance by allowing less variability of the filter frequency response; ultimately, leading to a 
	trade-off between the stochastic deviation robustness and the SGNN discrimination power. 
	
	\item \emph{Duality gap and convergence (Sections \ref{sec:dualityGap}-\ref{sec:Convergence})}: We analyze the optimality loss of the variance-constrained learning by characterizing the duality gap of the formulated optimization problem and the converging behavior of the proposed primal-dual algorithm. The sub-optimality is bounded proportionally by the representation capacity of the SGNN, the gradient descent approximation at the primal phase, and the gradient ascent step-size at the dual phase. These findings validate the effectiveness of the variance-constrained learning and identify our handle to obtain near-optimal solutions.
	
\end{enumerate}

This paper contains one additional minor contribution. It conducts theoretical analysis of stochastic graph filters and SGNNs with a more general stochastic graph model than earlier works, where a subset of edges are dropped with a probability $p$ and another subset are added with another probability $q$ [Def. \ref{def_res}]. The theoretical findings of this work are not presented in the preliminary version \cite{gao2021variance}, which focused on the algorithm. Numerical simulations on source localization and recommender systems corroborate the theoretical findings in Section \ref{numer}. The conclusions are drawn in Section \ref{sec_conclusions}. All proofs and lemmas used in these proofs are are collected in the appendix. 

\section{Stochastic Graph Neural Network} \label{sec:StoGNN}




Let $\mathcal{G}= ( \mathcal{V}, \mathcal{E}, \bbS)$ be a graph with node set $\mathcal{V} = \{ 1,\cdots,n \}$, edge set $\ccalE = \{(i,j)\} \subseteq \ccalV \times \ccalV$, and graph shift operator $\bbS \in \mathbb{R}^{n \times n}$, e.g., the adjacency matrix $\bbA$ or the Laplacian matrix $\bbL$. Let also $\bbx = [ x_1,...,x_n ]^\top \in \mathbb{R}^n$ be a graph signal with component $x_i$ the signal value associated to node $i$ \cite{Shuman2013, Segarra2017, ortega2018, Coutino2019}. For example, in a recommender system nodes are movies, edges are similarities between them, and the graph signal is the ratings given by a user to these movies. We are interested in learning representations from the tuple ($\ccalG, \bbx$) for tasks such as inferring user missing ratings, while we aim to keep these representations robust w.r.t. random topological changes on the nominal graph. These random changes may be due to different factors such as adversarial attacks \cite{zugner2018adversarial}, communication link outage \cite{guo2015outage}, and edge rewiring in collaborative filtering to improve diversity \cite{par2004recommender}. In these cases, existing edges may be lost and new edges may be added, resulting in random topologies. We characterize the latter with the generalized random edge sampling (GRES) model.

\begin{definition}[GRES($p, q$) model]\label{def_res} 
	Consider the nominal graph $\ccalG = (\ccalV, \ccalE)$. Let $\ccalE_d \subseteq \ccalE$ be a set of $M_d$ existing edges that may be dropped and $\ccalE_a \nsubseteq \ccalE$ a set of $M_a$ new edges that may be added. A GRES graph realization $\ccalG_k = (\ccalV, \ccalE_k)$ of $\ccalG$ comprises the same node set $\ccalV$ and the edge set $\ccalE_k$ where the edges in $\ccalE_d$ are dropped independently with a probability $0 \le p < 1$ and the edges in $\ccalE_a$ are added independently with a probability $0\le q < 1$.
\end{definition}\vspace{-1mm}
\noindent We denote by $\bbS_k$ the random shift operator of the GRES($p,q$) graph $\ccalG_k$ with $2^{M_d + M_a}$ possible realizations.

\smallskip
\noindent\textbf{Stochastic graph neural network (SGNN) \cite{gao2021stochastic}.} An SGNN is a graph neural network that learns representations over random topologies. The key of this architecture is the \emph{stochastic graph filter}. When applied to a graph signal $\bbx$, the output of a stochastic graph filter over a sequence of $K$ GRES($p,q$) graph realizations $\{ \bbS_k \}_{k=0}^K$ can be written as \vspace{-2mm}
\begin{equation} \label{eq:stochasticGraphFilterOutput}
	\begin{split}
		\bbH(\bbS_{K:0})\bbx :=  \sum_{k=0}^K h_{k}\bbS_k \ldots \bbS_1 \bbS_0 \bbx =  \sum_{k=0}^K h_{k} \prod_{i=0}^k \bbS_{i} \bbx
	\end{split}
\end{equation}
with $\{ h_k \}_{k=0}^K$ the filter coefficients and $\bbS_0 \!=\! \bbI$ the identity matrix \cite{gao2021stochastic}. In the filter output \eqref{eq:stochasticGraphFilterOutput}, the first shift $\bbS_1 \bbx$ collects at each node the information from its immediate neighbors and the successive $k$-shifts $\prod_{i=0}^k \bbS_{i} \bbx$ collect information from $k$-hop neighbors that can be reached via the randomly present edges in $\bbS_1,\ldots,\bbS_k$. The stochastic graph filter aggregates these shifted signals $\{ \prod_{i=0}^k \bbS_{i} \bbx \}_{k=0}^K$ and weighs them with coefficients $\{h_k\}_{k=0}^K$; ultimately, allowing for a distributed implementation -- see also \cite{Segarra2017, Shuman2018}.

An SGNN is a layered architecture, in which each layer comprises a bank of stochastic graph filters followed by a pointwise nonlinearity. At layer $\ell =1,..., L$, the input is a collection of $F$ graph signal features $\{ \bbx_{\ell-1}^g \}_{g=1}^F$ generated at the former layer $\ell-1$. These features are processed by a bank of $F^2$ stochastic graph filters $\{ \bbH_{\ell}^{fg}(\bbS_{K:0})\}_{fg}$ [cf. \eqref{eq:stochasticGraphFilterOutput}], aggregated over the input index $g$, and finally passed through a nonlinearity $\sigma(\cdot)$ to generate $F$ output features of layer $\ell$, i.e., \vspace{-2mm}
\begin{equation}\label{eq:sgnn}
	\begin{split}
		\bbx_{\ell}^{f} \!=\! \sigma \Big( \sum_{g=1}^{F} \bbu_{\ell}^{fg} \Big)\!\!=\! \sigma\Big(\sum_{g=1}^{F} \bbH_{\ell}^{fg}(\bbS_{K:0})\bbx_{\ell-1}^g\! \Big),~\!\for~f\!=\!1,...,F.
	\end{split}
\end{equation}
To ease exposition, we consider a single input $\bbx_0^1 = \bbx$ and output $\bbx_L^1$. We represent the SGNN as the nonlinear map $\bbPhi(\cdot;\bbS_{P:1},\ccalH): \mathbb{R}^n \to \mathbb{R}^n$, which applies on the input $\bbx$ and generates the output $\bbPhi(\bbx;\bbS_{P:1},\ccalH):=\bbx_L^1$. Here, $\ccalH = \{ h_{0\ell}^{fg}, \ldots. h_{K\ell}^{fg} \}_{fg\ell}$ collects all filter coefficients and $\bbS_{P:1}$ indicates the sequence of all $P=K[2F + (L-1)F^2]$ shift operators in the SGNN.

\smallskip
\noindent \textbf{Problem motivation.} The SGNN output $\bbPhi(\bbx;\bbS_{P:1},\ccalH)$ is a random variable because of the graph stochasticity and the data distribution. Given a training set $\ccalT= \{ (\bbx,\bby) \}$ and a loss function $\ccalC(\cdot,\cdot)$, we train the SGNN with stochastic gradient descent and the latter is shown equivalent to solving an unconstrained stochastic optimization problem over the graph and the data distributions \cite{gao2021stochastic}, i.e., 
\begin{equation}\label{eq:objective}
	\begin{split}
		& \mathbb{P}_{\text{un}}:=\min_{\ccalH} \mathbb{E}_\ccalM\! \left[ \ccalC_\ccalT(\bby, \bbPhi(\bbx;\bbS_{P:1},\ccalH)) \right]
	\end{split}
\end{equation}
where $\ccalC_\ccalT(\bby, \bbPhi(\bbx;\bbS_{P:1},\ccalH)):=\mathbb{E}_\ccalT\! \left[ \ccalC(\bby, \bbPhi(\bbx;\bbS_{P:1},\ccalH))\right]$ is the expected cost over the data distribution and $\ccalM$ is the discrete set of the shift operator sequences $\bbS_{P:1}$, which contains $2^{P (M_d+M_a)}$ elements. The expectation of a function $f(\bbx;\bbS_{P:1})$ over $\ccalM$ is $\mathbb{E}_\ccalM = \sum_{\bbS_{P:1}\in \ccalM} f(\bbx;\bbS_{P:1}) \mu(\bbS_{P:1})$ 
where $\mu(\cdot)$ is the probability measure over $\ccalM$ such that $\mu(\bbS_{P:1}) = 1/2^{P (M_d+M_a)}$ for each $\bbS_{P:1} \in \ccalM$. The solution of \eqref{eq:objective} accounts for the graph stochasticity during training and makes the SGNN robust when tested over random graphs. However, problem \eqref{eq:objective} only guarantees robustness w.r.t. the expected performance but ignores stochastic deviations around it. The latter may lead to a single SGNN output far from the optimal expectation and be problematic in settings where uncertainty must be controlled.

To overcome this issue, we propose a variance-constrained learning strategy for the SGNN to balance the expected performance with stochastic deviations. Specifically, we formulate a constrained stochastic optimization problem as
\begin{alignat}{3} \label{eq:varianceConstrainedProblem}
	\mathbb{P}_{\text{con}}:=   &\min_{\ccalH} \mathbb{E}_\ccalM \left[ \ccalC_\ccalT(\bby, \bbPhi(\bbx;\bbS_{P:1},\ccalH)) \right]             \\
	&  \st \quad {\rm Var} \left[ \bbPhi(\bbx;\bbS_{P:1},\ccalH) \right] \le C_v  \nonumber
\end{alignat}
where ${\rm Var}[\bbPhi(\bbx;\bbS_{P:1},\ccalH)]$ is a variance measure that characterizes stochastic deviations of the SGNN output and $C_v$ is a variance bound we can tolerate. Problem \eqref{eq:varianceConstrainedProblem} is challenging because of the non-convexity of the SGNN, the stochasticity of the GRES($p,q$) model, and the variance constraint. We solve the problem via a primal-dual learning method in Sec. \ref{sec_problem}. Since the proposed variance-constrained learning trades the variance with the discrimination power, we characterize this trade-off explicitly and show the role played by different factors in Sec. \ref{DOL}. We further analyze the optimality loss induced by the primal-dual method in Sec. \ref{sec:dualityGap} and prove this learning procedure converges to a neighborhood of the saddle point solution in Sec. \ref{sec:Convergence}.

\section{Variance-Constrained Learning}\label{sec_problem}


We consider the average variance experienced over all nodes
\begin{align} \label{eq:varianceDefinition}
		&{\rm Var} \left[ \bbPhi(\bbx;\bbS_{P:1},\ccalH) \right] := \frac{1}{n}\sum_{i=1}^n {\rm Var}\Big[[\bbPhi(\bbx;\bbS_{P:1},\ccalH)]_i\Big]\\
		&=\! \frac{1}{n}\!\sum_{i=1}^n \!\!\Big(\!\mathbb{E}_{\ccalM}\!\Big[[\bbPhi(\bbx;\bbS_{P:1},\ccalH)]^2_i\Big] \!\!-\! \mathbb{E}_{\ccalM}\!\Big[[\bbPhi(\bbx;\bbS_{P:1},\ccalH)]_i\Big]^2\Big). \nonumber
\end{align}
This expression measures how individual node outputs $\{[\bbPhi(\bbx;\bbS_{P:1},\ccalH)]_i\}_{i=1}^n$ deviate from their expectations. It is a standard criterion used in multi-dimensional systems and is related to the A-optimality of the confidence ellipsoid \cite{joshi2008sensor}. In what follows, we use \eqref{eq:varianceDefinition} as the variance measure in \eqref{eq:varianceConstrainedProblem} and solve the latter problem with a primal-dual learning procedure. We further show how this learning strategy behaves as a \emph{self-learning variance regularizer} that provides explicit theoretical guarantees about stochastic deviations.

Since problem \eqref{eq:varianceConstrainedProblem} is a constrained optimization problem, we solve it in the dual domain. However, the variance constraint is a \emph{non-convex} function of $\mathbb{E}_{\ccalM}\big[\bbPhi(\bbx;\bbS_{P:1},\ccalH)\big]$ and  $\mathbb{E}_{\ccalM}\big[\bbPhi(\bbx;\bbS_{P:1},\ccalH)^2\big]$. The latter makes it difficult to analyze the duality gap, which quantifies the optimality loss of the solution obtained in the dual domain; consequently, there is no performance guarantee for any dual method solving \eqref{eq:varianceConstrainedProblem} as we shall detail in Sec. \ref{sec:dualityGap}. To provide theoretical performance guarantees, we consider the surrogate problem where we constrain separately the first and second order moments in \eqref{eq:varianceDefinition}, i.e.,
\begin{alignat}{3} \label{eq:alternativeVarianceConstrainedProblem}
	\mathbb{P}:=   &\min_{\ccalH} \mathbb{E}_\ccalM \left[ \ccalC_\ccalT(\bby, \bbPhi(\bbx;\bbS_{P:1},\ccalH)) \right]             \\
	&\st \frac{1}{n}\mbE_\ccalM \Big[ \sum_{i=1}^n [\bbPhi(\bbx;\bbS_{P:1},\ccalH)]_i \Big] \ge C_f,  \nonumber \\
	&\quad ~~ \frac{1}{n}\mbE_\ccalM \Big[ \sum_{i=1}^n [\bbPhi(\bbx;\bbS_{P:1},\ccalH)]^2_i \Big] \le C_s. \nonumber
\end{alignat}
The constraints of \eqref{eq:alternativeVarianceConstrainedProblem} are \emph{convex} functions (the outer function not the composed function with the SGNN) of the first order moment $\mathbb{E}_{\ccalM}\big[\bbPhi(\bbx;\bbS_{P:1},\ccalH)\big]$ and of the second order moment $\mathbb{E}_{\ccalM}\big[\bbPhi(\bbx;\bbS_{P:1},\ccalH)^2\big]$, respectively\footnote{A more intuitive constraint for the first order moment is to lower bound its absolute value $\big|\mathbb{E}_{\ccalM}\big[\bbPhi(\bbx;\bbS_{P:1},\ccalH)\big]\big| \ge C_f$, i.e., $C_f - \big|\mathbb{E}_{\ccalM}\big[\bbPhi(\bbx;\bbS_{P:1},\ccalH)\big]\big| \le 0$. However, the latter is still a non-convex function and thus does not allow for the duality gap analysis as \eqref{eq:alternativeVarianceConstrainedProblem}.}.  
Through scalar $C_f \ge 0$ we lower bound the expected output and through scalar $C_s \ge 0$ we upper bound the output autocorrelation. The latter are related to the variance \eqref{eq:varianceDefinition}; hence, we can implicitly bound the variance as 
\begin{align}\label{eq:resultConstrainedVariance}
	{\rm Var} \left[ \bbPhi(\bbx;\bbS_{P:1},\ccalH) \right] \le C_s - C_f^2.
\end{align}
Since there always exist $C_f$ and $C_s$ such that $C_s - C_f^2 = C_v$, e.g., $C_f = 0$ and $C_s = C_v$, the surrogate problem \eqref{eq:alternativeVarianceConstrainedProblem} restricts the SGNN output and balances the expected performance with the stochastic deviation as the original problem \eqref{eq:varianceConstrainedProblem}. 

\subsection{Primal-Dual Learning}\label{subsec:varianceConstrained}

By introducing the non-negative dual variable $\bbgamma \!=\! [\gamma_1,\gamma_2] \in \mathbb{R}^2_+$, we define the Lagrangian $\ccalL(\ccalH, \!\bbgamma)$ of \eqref{eq:alternativeVarianceConstrainedProblem} as
\begin{align} \label{eq:Lagrangian}
	\ccalL(\ccalH, \bbgamma) &= \mathbb{E}_\ccalM\! \big[\ccalC_\ccalT\big(\bby, \bbPhi(\bbx;\bbS_{P:1},\ccalH)\big)\big]\nonumber\\ 
	&+ \gamma_1 \Big( C_f - \frac{1}{n}\mbE_{\ccalM} \Big[ \sum_{i=1}^n [\bbPhi(\bbx;\bbS_{P:1},\ccalH)]_i \Big]\Big) \\
	& - \gamma_2 \Big( C_s - \frac{1}{n}\mbE_{\ccalM} \Big[ \sum_{i=1}^n [\bbPhi(\bbx;\bbS_{P:1},\ccalH)]^2_i \Big] \Big).\nonumber
\end{align}
Given the dual function $\mathcal{D}(\bbgamma) = \min_{\ccalH} \mathcal{L}(\ccalH,\bbgamma)$, it holds that $\mathcal{D}(\bbgamma) \le \mathbb{P}$ for any $\bbgamma$ \cite{nocedal2006numerical}. The goal now is to find the optimal dual variable $\bbgamma^*$ that maximizes the dual function as
\begin{equation} \label{eq:dualFunction}
	\begin{split}
		\mathbb{D} = \max_{\bbgamma} \mathcal{D}(\bbgamma):=\max_{\bbgamma} \min_{\ccalH} \mathcal{L}(\ccalH,\bbgamma).
	\end{split}
\end{equation}
That is, search for an optimal primal-dual pair $(\ccalH^*, \bbgamma^*)$ satisfying the saddle-point relationship $\mathcal{L}(\ccalH^*,\bbgamma) \le \mathcal{L}(\ccalH^*,\bbgamma^*) \le \mathcal{L}(\ccalH,\bbgamma^*)$ 
for any $\ccalH$ and $\bbgamma$ in the neighborhood of the optimal solution.

We approach the dual problem \eqref{eq:dualFunction} by alternatively updating the primal variable $\ccalH$ with stochastic gradient descent and the dual variable $\bbgamma$ with stochastic gradient ascent. 

\smallskip
\noindent\emph{Primal phase.} At iteration $t$, given the primal variable $\ccalH_t$ and the dual variable $\bbgamma_t$, we set $\ccalH^{(0)}_{t} = \ccalH_t$ and update the primal variable with gradient descent for $\Gamma$ steps as
\begin{subequations}\label{eq_priup}
\begin{align}
	&\ccalH^{(\tau)}_t \!=\! \ccalH^{(\tau\!-\!1)}_{t} \!\!- \eta_\ccalH \nabla_{\ccalH} \mathcal{L}(\ccalH^{(\tau\!-\!1)}_{t}\!\!,\bbgamma_{t}),~\text{for}~\tau \!=\! 1,...,\Gamma,\\
	&\ccalH_{t+1} := \ccalH^{(\Gamma)}_t
\end{align}
\end{subequations}
where $\eta_\ccalH>0$ is the primal step-size. The challenge in \eqref{eq_priup} is to compute the gradient $\nabla_{\ccalH} \mathcal{L}(\ccalH^{(\tau-1)}_{t},\bbgamma_{t})$, which requires evaluating the expectation $\mathbb{E}_{\ccalM}[\cdot]$. The latter needs to be estimated over $2^{P(M_d+M_a)}$ realizations resulting in an expensive computation. To overcome this issue, we approximate the expectation with empirical alternatives over $N$ 
sampled realizations $\{ \bbS^{(j)}_{P:1}\}_{j=1}^N$ as
\begin{subequations} \label{eq:averageApproximate1}
\begin{align} \label{eq:averageApproximate15}
	&\mathbb{E}_\ccalM\! \big[\ccalC_\ccalT(\bby,\! \bbPhi(\bbx;\bbS_{P:1},\!\ccalH))\big] \!\approx\! \frac{1}{N}\!\sum_{j=1}^N \ccalC_\ccalT(\bby,\! \bbPhi(\bbx;\bbS_{P:1}^{(j)},\ccalH)), \\
	\label{eq:averageApproximate2}&\mbE_\ccalM\! \Big[ \sum_{i=1}^n [\bbPhi(\bbx;\bbS_{P:1},\ccalH)]_i \Big]\!\approx\! \frac{1}{N}\!\sum_{j=1}^N \sum_{i=1}^n [\bbPhi(\bbx;\bbS^{(j)}_{P:1},\ccalH)]_i,\\
	\label{eq:averageApproximate3}&\mbE_\ccalM\! \Big[ \sum_{i=1}^n [\bbPhi(\bbx;\bbS_{P:1},\ccalH)]^2_i \Big]\!\approx\! \frac{1}{N}\!\sum_{j=1}^N\sum_{i=1}^n [\bbPhi(\bbx;\bbS^{(j)}_{P:1},\ccalH)]^2_i.
\end{align}
\end{subequations}
The sampling average is a standard procedure in stochastic optimization methods, such as Monte-Carlo simulation \cite{harrison2010introduction} and stochastic gradient descent \cite{gao2022balancing}. 
A larger $N$ approximates better the expectation but in turn results in more computations, which yields a trade-off between the performance and complexity. If the problem dimension increases, the variance may be harder to approximate and we may increase $N$ to improve performance; if the problem dimension decreases, the variance may be easier to approximate and we may decrease $N$ to save computation. We shall show in Sec. \ref{numer} that for a graph of $50$ nodes, an $N \ge 10$ is sufficient.

\smallskip
\noindent\emph{Dual phase.} Given the updated primal variable $\ccalH_{t+1}$, the dual variable is updated with gradient ascent 
\begin{subequations} \label{eq:dualUpdate}
\begin{align} \label{eq_dualup1}
	&\gamma_{1, t+1} \!=\!\Big[ \gamma_{1, t} \!+\! \eta_\gamma  \Big( C_f \!-\! \frac{1}{n}\mbE_\ccalM\! \Big[ \sum_{i=1}^n [\bbPhi(\bbx;\bbS_{P:1},\ccalH)]_i \Big] \Big) \Big]_+,\\
	\label{eq_dualup2} &\gamma_{2, t+1} \!=\!\Big[ \gamma_{2, t} \!-\! \eta_\gamma  \Big( C_s \!-\!\frac{1}{n} \mbE_\ccalM\! \Big[ \sum_{i=1}^n [\bbPhi(\bbx;\bbS_{P:1},\ccalH)]^2_i \Big] \Big) \Big]_+
\end{align}
\end{subequations}
where $\eta_\gamma>0$ is the dual step-size and $[\cdot]_+$ is the non-negative projection since $\gamma_1,\gamma_2 \ge 0$. In \eqref{eq_dualup1} and \eqref{eq_dualup2}, we substitute the expectations with their empirical alternatives as in \eqref{eq:averageApproximate2} and \eqref{eq:averageApproximate3}. These stochastic approximations allow updating the dual step and completing the iteration $t$. The algorithm is stopped either after a maximum number of iterations $T$ or when a tolerance on the gradient norm is reached. Algorithm \ref{alg:primalDualAlgorithm} recaps this procedure. 

{\linespread{1}
	\begin{algorithm}[t] \begin{algorithmic}[1] 
			\STATE \textbf{Input:} Training set $\ccalT$, loss function $\ccalC(\cdot,\cdot)$, initial primal variable $\ccalH_0$, initial dual variable $\bbgamma_0$, bounds $C_f$ and $C_s$, primal step-size $\eta_\ccalH$, and dual step-size $\eta_\gamma$
			\STATE Establish the Lagrangian \eqref{eq:Lagrangian} and the dual problem \eqref{eq:dualFunction}
			\FOR {$t = 0,1,2,...$}
			\STATE \textbf{Primal phase.} Given $\ccalH_t$ and $\bbgamma_t$, update the primal variable with gradient descent for $\Gamma$ steps [cf. \eqref{eq_priup}]\\
			\STATE Approximate $\mathcal{L}(\ccalH^{(\tau\!-\!1)}_{t}\!\!,\bbgamma_{t})$ stochastically [cf. \eqref{eq:averageApproximate1}]\\
			\STATE \textbf{Dual phase.} Given $\ccalH_{t+1}$ and $\bbgamma_t$, update the dual variable with stochastic gradient ascent [cf. \eqref{eq:dualUpdate}]\\
			\ENDFOR			
		\end{algorithmic}
		\caption{Primal-dual learning procedure}\label{alg:primalDualAlgorithm}
\end{algorithm}}

\begin{remark}
	Algorithm \ref{alg:primalDualAlgorithm} is applicable to both the original problem \eqref{eq:varianceConstrainedProblem} and the surrogate problem \eqref{eq:alternativeVarianceConstrainedProblem}. We focus on the surrogate problem \eqref{eq:alternativeVarianceConstrainedProblem} because it allows to analyze its duality gap in Sec. \ref{sec:dualityGap}; hence, providing a unified exposition throughout the paper. However, if the duality analysis is not of interest and any local minima is acceptable, we can work with the original problem \eqref{eq:varianceConstrainedProblem} directly. All the other theoretical findings -- the above primal-dual learning, the discrimination analysis in Sec. \ref{DOL} and the convergence analysis in Sec. \ref{sec:Convergence} -- apply to the original problem as well.
\end{remark} 
\begin{remark}
	Any stochastic optimization algorithm can be used at the primal phase to solve the dual function $\min_\ccalH \ccalL(\ccalH, \bbgamma)$ [cf. \eqref{eq:dualFunction}]. We apply the stochastic gradient descent in \eqref{eq_priup} as a baseline method to ease the exposition. Other choices include the ADAM method, the quasi-Newton method, etc. 
\end{remark} 

\subsection{Self-Learning Variance Regularizer}\label{subsec:varianceRegularized}

An intuitive alternative to the variance-constrained problem \eqref{eq:varianceConstrainedProblem} is to consider the variance as a regularizer for problem \eqref{eq:objective}, i.e.,
\begin{align} \label{eq:linktoRegularizer}
	\min_{\ccalH} \mathbb{E}_\ccalM\! \big[\ccalC_\ccalT(\bby,\bbPhi(\bbx;\bbS_{P:1},\ccalH))\big] + \beta {\rm Var} \left[ \bbPhi(\bbx;\bbS_{P:1},\ccalH) \right]
\end{align}
where $ \beta > 0$ is the regularization parameter. The regularization term $ \beta {\rm Var} \left[ \bbPhi(\bbx;\bbS_{P:1},\ccalH) \right]$ incentivizes the SGNN output to have a small variance by forcing its parameters to trade between the expected cost and the variance. Problem \eqref{eq:linktoRegularizer} can be solved directly with stochastic gradient descent. However, we find it limiting in two aspects: (i) It does not provide theoretical guarantees for stochastic deviations. The explicit relation between the regularization term and the stochastic deviation is unclear, thus little insights or implications can be obtained; (ii) It is difficult to select a suitable regularization parameter $\beta$ that well balances the expected performance and the variance. If $\beta$ is too large, the SGNN would only restrict the variance but sacrifice the performance; if $\beta$ is too small, the SGNN may generate outputs with a large variance. Deciding the value of $\beta$ requires extensive cross-validation and could be computationally demanding.

Differently, the variance-constrained learning not only optimizes the SGNN parameters $\ccalH$ akin to the variance regularized objective, but also learns the regularization parameter $\bbgamma$ based on the variance bound. To see this, recall that minimizing the Lagrangian \eqref{eq:Lagrangian} at the primal phase is equivalent to solving \vspace{-2mm}
\begin{align} \label{eq:varianceLagrangian}
	 \!\!\min_\ccalH &\mathbb{E}_\ccalM\! \big[\ccalC_\ccalT(\bby,\! \bbPhi(\bbx;\!\bbS_{P:1},\!\ccalH))\big]\!\!-\! \frac{\gamma_1}{n}\mbE_{\ccalM}\! \Big[\! \sum_{i=1}^n [\bbPhi(\bbx;\!\bbS_{P:1},\!\ccalH)]_i \!\Big]\nonumber \\
	& ~~~~~~~~ ~~~~~~~~~~~~ + \frac{\gamma_2}{n} \mbE_{\ccalM} \Big[ \sum_{i=1}^n [\bbPhi(\bbx;\bbS_{P:1},\ccalH)]^2_i \Big].
\end{align}
This is similar to the variance regularizer in \eqref{eq:linktoRegularizer}, where the dual variable $\bbgamma = [\gamma_1,\gamma_2]$ is the regularization parameter and the primal variable $\ccalH$ is updated in the direction that reduces the variance [cf. \eqref{eq:varianceDefinition}]. However, instead of hand-fixing $\bbgamma$ at the outset, the variance-constrained learning updates $\bbgamma$ at the dual phase based on the bounds of the first and second order moments $C_f, C_s$ [cf. \eqref{eq:dualUpdate}]; ultimately, based on the variance bound $C_v$ [cf. \eqref{eq:resultConstrainedVariance}]. Hence, we can consider the latter as a \emph{self-learning variance regularizer}, where the regularization parameter is learned based on the variance bound $C_v$.

More importantly, feasible solutions of the variance-constrained problem provide explicit theoretical guarantees about stochastic deviations of the SGNN output around its expectation. The following proposition establishes the probability contraction bound for the SGNN output and the role of the variance constraint. \vspace{-1mm}
\begin{proposition}\label{thm:probabilityBound}
	Consider the variance-constrained problem \eqref{eq:varianceConstrainedProblem}. Let $\ccalH$ be a feasible solution that satisfies the variance constraint. Then, for any $\varepsilon>0$, it holds that \vspace{-2mm}
	\begin{equation} \label{eq:ChebyshevBound}
		\begin{split}
			\text{Pr} \Big( \frac{1}{n}\big\| \bbPhi(\bbx;\bbS_{P:1},\!\ccalH) \!-\! \mathbb{E}_\ccalM[\bbPhi(\bbx;\bbS_{P:1},\!\ccalH)] \big\|^2 \!\le\! \varepsilon \!\Big) \!\ge\! 1\!-\! \frac{C_v}{\varepsilon}.\nonumber
		\end{split}
	\end{equation}
\end{proposition} 

\begin{proof}
	See Appendix \ref{pr:proofTheorem1}.
\end{proof} 

\noindent That is, the probability that an SGNN realization deviates from its expectation by at most $\varepsilon$ is no more than a fraction of $C_v/{\varepsilon}$. When the variance constraint is strict, i.e., $C_v \to 0$, the bound approaches one and stochastic deviations are well-controlled, but it may be challenging to find a feasible solution. 
The result shows an explicit relation between the imposed variance constraint and random SGNN behaviors, which cannot be established by the variance regularizer in \eqref{eq:linktoRegularizer}. 

\section{Variance and Discrimination} \label{DOL}


Compared to the unconstrained problem \eqref{eq:objective}, 
problem \eqref{eq:varianceConstrainedProblem} trades the bounded variance with the expected performance. However, the explicit trade-off is unclear, i.e., how the imposed constraint affects the overall performance. To address the latter, we characterize theoretically the variance of the SGNN output and show that the variance-constrained learning improves the robustness to stochastic deviations by shrinking the frequency response of stochastic graph filters [cf. \eqref{eq:stochasticGraphFilterOutput}] within the SGNN; thus, reducing the discrimination power. To obtain this result, we analyze next the SGNN behaviors in the graph spectral domain. 

\subsection{Frequency Response of Stochastic Graph Filter}

Consider the shift operator eigendecomposition $\bbS_k = \bbV_k \bbLambda_k \bbV_k^{\top}$ with eigenvectors $\bbV_k = [\bbv_{k1}, \cdots, \bbv_{kn}]$ and eigenvalues $\bbLambda_k = \text{diag} ([\lambda_{k1},...,\lambda_{kn}])$. The graph Fourier transform (GFT) is the projection of signal $\bbx$ onto $\bbV_k$, i.e., $\bbx = \sum_{i=1}^n \hat{x}_i \bbv_{ki}$, where $\hat{\bbx} = [\hat{x}_i, \cdots, \hat{x}_n]^\top$ are the Fourier coefficients \cite{ortega2018}. Given the eigendecompositions of $k+1$ successive shift operators $\bbS_0,...,\bbS_k$, we can perform a chain of GFTs on $\bbx$ as 
\begin{equation}\label{eq:kShift}
	\prod_{i=0}^k \bbS_{i}\bbx \!=\!\sum_{i_0=1}^n\! \sum_{i_1=1}^n\!\! \cdots\!\! \sum_{i_k=1}^n\! \hat{x}_{0i_0}\hat{x}_{1i_{0}i_1}\cdots \hat{x}_{ki_{k-1}i_k} \prod_{j=0}^k\lambda_{ji_j} \bbv_{ki_k}
\end{equation}
for all $k = 0, ..., K$, where we first perform the GFT over $\bbS_0$, then over $\bbS_1$, and so on. Here, $\{ \hat{x}_{0i_0} \}_{i_0=1}^n$, $\{\hat{x}_{ji_{j-1}i_j}\}_{j=1}^k$ are the Fourier coefficients of expanding $\bbx$ on the chain of $\bbS_0,...,\bbS_k$ -- see also \cite{gao2021stochastic, gao2021stability}. Thus, we can represent the filter output $\bbu \!=\! \bbH(\bbS_{K:0})\bbx$ as 
\begin{align}\label{eq:GFTChain}
	\bbu \!=\!\!\sum_{i_0=1}^n\! \sum_{i_1=1}^n \!\!\cdots\!\!\! \sum_{i_K=1}^n \!\!\!\hat{x}_{0i_0}\hat{x}_{1i_{0}i_1}\cdots \hat{x}_{Ki_{K\!-\!1}i_K}\!\! \sum_{k=0}^K\!\! h_k\!\! \prod_{j=0}^k\!\lambda_{ji_j}\! \bbv_{Ki_K}\!.
\end{align} 
As it follows from \eqref{eq:GFTChain}, the input-output relation of the filter in the spectral domain is determined by the eigenvalues $\bbLambda_K,\ldots, \bbLambda_1$ and eigenvectors $\bbV_K,\ldots,\bbV_1$. We can then define the \emph{frequency response of the stochastic graph filter} as 
\begin{equation} \label{eq:generalizeFrequencyResponse}
	\begin{split}
		h(\bblambda) := \sum_{k=0}^K h_k \prod_{j=0}^k \lambda_{j}
	\end{split}
\end{equation}
which is a $K$-dimensional analytic function of the generic frequency vector variable $\bblambda = [\lambda_1,\ldots,\lambda_K]^\top \in \mathbb{R}^K$ with $\lambda_0=1$ by default (i.e., $\bbS_0 = \bbI$) \cite{gao2021stochastic}. The frequency response $h(\bblambda)$ is a multivariate function of a $K$-dimensional vector variable $\bblambda$, where the $k$th entry $\lambda_k$ is the analytic variable corresponding to the $k$th shift operator $\bbS_k$. The shape of $h(\bblambda)$ is determined by the coefficients $\{h_k\}_{k=0}^K$, while a specific chain of $\bbS_K,\ldots,\bbS_1$ only instantiates the eigenvalues $\{\lambda_{Ki}\}_{i=1}^n,\ldots,\{\lambda_{1i}\}_{i=1}^n$ on the $K$-dimensional variable $\bblambda$ -- see Fig. \ref{fig:generalizedFrequencyResponse} for an example. 

\subsection{Variance Analysis}

Given the filter frequency response over stochastic graphs, we make the following conventional assumptions. 
\begin{assumption} \label{as:LipschitzFilter}
	Let $h(\bblambda)$ be the filter frequency response [cf. \eqref{eq:generalizeFrequencyResponse}] of the $K$-dimensional variable $\bblambda$ satisfying $|h(\bblambda)| \le 1$. The stochastic graph filter is Lipschitz, i.e., there exists a constant $C_L$ such that 
	\begin{equation} \label{eq:LipschitzFilter}
		\begin{split}
			\left|h(\bblambda_1) \!-\! h(\bblambda_2)\right| \!\le\! C_L \|\bblambda_1 \!-\! \bblambda_2\|,~\forall~\bblambda_1,\bblambda_2 \in \Lambda^K
		\end{split}
	\end{equation}
    where $\Lambda^K$ is the considered $K$-dimensional domain.
\end{assumption} 

\begin{assumption} \label{as:LipschitzNonlinearity1}
	The nonlinearity $\sigma(\cdot)$ satisfies $\sigma(0)\!=\!0$ and it is Lipschitz, i.e., there exists a constant $C_\sigma$ such that 
	\begin{equation}\label{eq:assumptionLipschitzNonlinearity1}
		\begin{split}
			|\sigma(x)-\sigma(y)| \le C_\sigma |x-y|,~\forall~x,y \in \mathbb{R}.
		\end{split}
	\end{equation}
\end{assumption} 

\begin{assumption} \label{as:LipschitzNonlinearity2}
The nonlinearity $\sigma(\cdot)$ is variance non-increasing, i.e., for any real random variable $x$, it holds that ${\rm Var}[\sigma(x)] \le {\rm Var}[x]$.
\end{assumption} 

Assumption \ref{as:LipschitzFilter} implies that the frequency response $h(\bblambda)$ does not change faster than linear in any frequency direction of $\bblambda$, which is standard in the stability analysis of GNNs \cite{Gama20-Stability}. It holds for filter coefficients $\{h_k\}_{k=0}^K$ and graph eigenvalues $\bblambda$ of finite values because $h(\bblambda)$ is a finite-order polynomial, such that it is bounded and Lipschitz for some $C_L < \infty$. Given $\{h_k\}_{k=0}^K$, we can express $h(\bblambda)$ and estimate $C_L$ as the maximal finite difference in the considered domain. Assumptions \ref{as:LipschitzNonlinearity1} and \ref{as:LipschitzNonlinearity2} hold for popular nonlinearities such as the ReLU and the absolute value \cite[Lemma 1]{gao2021stochastic}. The following theorem then formalizes the SGNN output variance. 

\begin{figure}[t]
	\centering
	\includegraphics[width=0.7\linewidth , height=0.5\linewidth, trim=50 50 50 50]{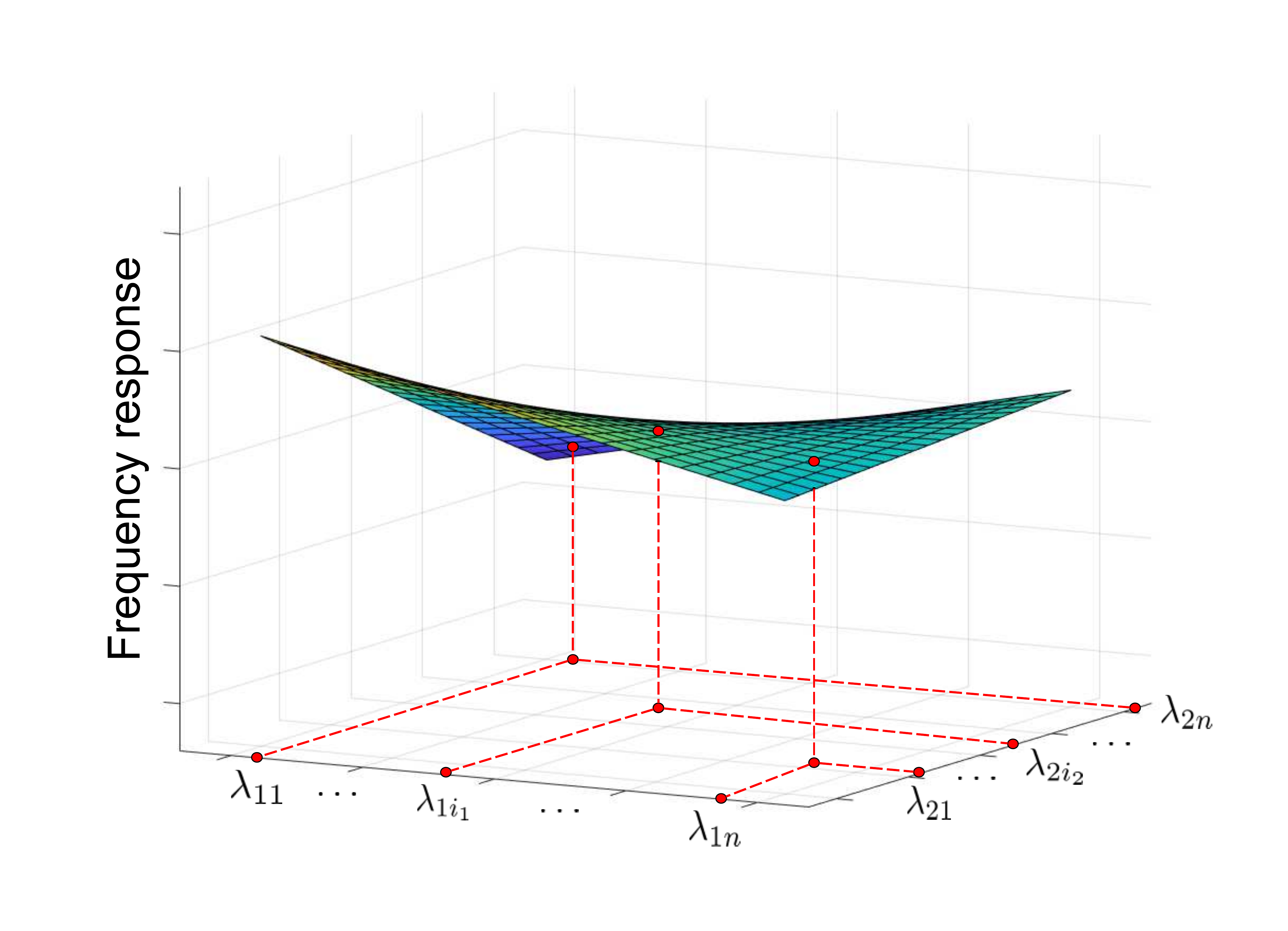}
	\caption{The $2$-dimensional frequency response of a stochastic graph filter. Function $h(\bblambda)$ is independent of graph realizations and it is completely defined by parameters $\{h_k\}_{k=0}^K$ [cf. \eqref{eq:generalizeFrequencyResponse}]. For a specific chain of graph realizations $\{\bbS_1,\bbS_2\}$, $h(\bblambda)$ is instantiated on specific eigenvalues $\{\lambda_{11},...,\lambda_{1n}\}$ determined by $\bbS_1$ and $\{\lambda_{21},...,\lambda_{2n}\}$ determined by $\bbS_2$.}\label{fig:generalizedFrequencyResponse} 
\end{figure}
\begin{theorem} \label{thm:varianceAnalysis}
	Consider the SGNN in \eqref{eq:sgnn} of $L$ layers, $F$ features, and filter order $K$ over the GRES($p,q$) model with $M_d$ dropping edges and $M_a$ adding edges [Def. \ref{def_res}]. Let the stochastic graph filters with the frequency response \eqref{eq:generalizeFrequencyResponse} satisfy Assumption \ref{as:LipschitzFilter} with $C_L$ and the nonlinearity $\sigma(\cdot)$ satisfy Assumptions \ref{as:LipschitzNonlinearity1}-\ref{as:LipschitzNonlinearity2} with $C_\sigma$. Then, for any input graph signal $\bbx$, the variance of the SGNN output is upper bounded as 
	\begin{align}\label{eq:varianceAnalysis}
			\!{\rm Var}\! \left[ \bbPhi(\bbx;\bbS_{P:1},\ccalH) \right] &\!\le\! C_L^2 \big(M_dp(1\!-\!p) + M_a q(1\!-\!q)\big) C \| \bbx \|_2^2\nonumber \\
			&~~ + \mathcal{O}(p^2(1\!-\!p)^2) + \mathcal{O}(q^2(1\!-\!q)^2) 
	\end{align}
	where $C= 4 K \sum_{\ell=1}^L \!F^{2L-3} C_\sigma^{2\ell-2} / n$ is a constant.
\end{theorem}
\begin{proof}
	See Appendix \ref{pr:proofTheorem2}.
\end{proof} 
Theorem \ref{thm:varianceAnalysis} states that the variance of the SGNN output is upper bounded proportionally by the Lipschitz term $C_L^2$ and quadratically by the edge dropping $\backslash$ adding probability $p ~\backslash~ q$. The constant $C$ shows the role of the architecture, i.e., the number of features $F$, layers $L$ and nonlinearity $C_\sigma$. The result identifies three explicit factors that affect the variance:

\begin{enumerate}[(1)]
	
	\item \emph{Filter property.} The first term $C_L^2$ captures the variation of the filter frequency response $h(\bblambda)$. The variance decreases with the Lipschitz constant $C_L$, which is determined by the learned parameters $\ccalH$. A smaller $C_L$ implies the frequency response $h(\bblambda)$ changes slower in the spectral domain; thus, it is more stable to frequency deviations induced by the graph stochasticity and leads to a lower variance. However, this flatter response reduces the filter capacity to discriminate between nearby spectral features, i.e., the filter has similar responses for graph frequencies that are close to each other. The latter indicates an implicit trade-off between decreasing the variance and increasing the discrimination power.
	
	\item \emph{Graph stochasticity.} The second term $M_dp(1\!-\!p) + M_a q(1-q)$ represents the impact of the graph stochasticity. The variance decreases when the number of dropping edges $M_d$ or adding edges $M_a$ is small. The variance decreases also when edges are stable ($p ~\backslash~ q \to 0$) or highly unstable ($p ~\backslash~ q \to 1$). The latter is because the maximal uncertainty on an edge is for $p = q = 0.5$. Such graph stochasticity depends typically on external factors (e.g., interference, attacks) or design choices (e.g., graph dropout). 
	
	\item \emph{SGNN architecture.} The third term $4 K \sum_{\ell=1}^L F^{2L-3} C_\sigma^{2\ell-2}/n$ indicates the effect of the SGNN architecture, which is the consequence of the graph stochasticity propagating through the nonlinearity ($C_\sigma$), filter banks ($F$), and layers ($L$). First, $C_\sigma$ is typically one implying the non-expansivity of the nonlinearity. Second, the variance decreases with an architecture of less layers $L$ or less features $F$. Both imply less stochastic graph filters, interact with less random components, and ultimately result in a lower variance. However, such an architecture may have a limiting representation capacity.
	
\end{enumerate}

The aforementioned analysis indicates that we can constrain the variance in three ways: (1) reducing the Lipschitz constant $C_L$; (2) reducing the number of random edges $M_d, M_a$ or edge probabilities $p, q$; (3) reducing the architecture width $F$ and depth $L$. However, (2) and (3) are determined at the outset and cannot be controlled during training. This implies that the variance-constrained learning keeps the variance bounded by tuning parameters $\ccalH$ to lower the Lipschitz constant $C_L$ of the stochastic graph filters [As. \ref{as:LipschitzFilter}]. Consequently, the stochastic graph filters exhibit flatter frequency responses and restricting the variance comes at the expense of the discrimination power. From this perspective, the variance bound $C_v$ cannot be set too small; i.e., if $C_v$ is small, $C_L$ decreases yielding a flatter frequency response; hence, a lower discrimination power. This is an implicit trade-off we have to cope with for improving the SGNN robustness to stochastic deviations. We also note that Theorem \ref{thm:varianceAnalysis} extends the variance analysis in \cite{gao2021stochastic}, which is the particular case when all edges are only dropped with a probability $p$.

\begin{remark}
The bound in \eqref{eq:varianceAnalysis} may be loose when $M_d$, $M_a$ are large and the graph changes dramatically, i.e., $p ~\backslash~ q$ are around $0.5$, essentially because this bound holds uniformly for all graphs. However, this result still shows that the SGNN output variance is bounded and there is a trade-off between the stochastic deviation robustness and the architecture discrimination power. In turn, this indicates how the variance-constrained learning affects the performance, which mechanism inside the SGNN is mostly responsible, and which are our handle to reduce this bound (potentially the output variance).
\end{remark} 

\section{Duality Gap}\label{sec:dualityGap}


We solved problem \eqref{eq:alternativeVarianceConstrainedProblem} in the dual domain, where there exists a duality gap $\mathbb{P} - \mathbb{D}$ between the primal and dual solutions. The null duality gap can be achieved for convex problems, while problem \eqref{eq:alternativeVarianceConstrainedProblem} is typically non-convex. 
The latter makes it unclear how close is the dual solution $\mathbb{D}$ of \eqref{eq:dualFunction} to the primal solution $\mathbb{P}$ of \eqref{eq:alternativeVarianceConstrainedProblem}. In this section, we argue that the formulated problem could have a small duality gap despite its nonconvexity, which guarantees a small optimality loss caused by the dual transformation. To show such a result, we first consider a more general version of \eqref{eq:alternativeVarianceConstrainedProblem}, where we generalize the SGNN to an unparameterized function and the discrete set of shift operator sequences to a continuous set. Upon proving this generalized setting has a null duality gap, we then analyze the duality deviation induced by two generalizations and 
characterize the duality gap of problem \eqref{eq:alternativeVarianceConstrainedProblem}. 

\subsection{Problem Generalization}

We can consider the SGNN $\bbPhi(\bbx;\bbS_{P:1}, \ccalH)$ as a parameterized model of a function $f(\bbx;\bbS_{P:1})$ that takes as inputs a graph signal $\bbx$ and a discrete sequence of shift operators $\bbS_{P:1} \in \ccalM$ and generates representational features as outputs. Problem \eqref{eq:alternativeVarianceConstrainedProblem} considers the expected objective and constraints over the discrete set $\ccalM$. The latter can be extended to a continuous set $\widetilde{\ccalM}$ via the following $\varepsilon$-Borel set \cite{srivastava2008course}.
\begin{definition}[$\varepsilon$-Borel set]\label{def1:BorelSet0}
	For a shift operator $\bbS_k$, the $\varepsilon$-Borel set of $\bbS_k$ is 
	\begin{align}\label{eq:BorelSet}
		\ccalB_\varepsilon(\bbS_{k}) \!:=\! \{ \widetilde{\bbS}_k \!\in\! \mathbb{R}^{n \times n}\!\!:\! \| \widetilde{\bbS}_k \!-\! \bbS_k \| \!\le\! \varepsilon \},~\!\for~k\!=\!1,...,P
	\end{align}
    where $\| \cdot \|$ is the $\ell_2$-norm.
\end{definition}  

\noindent The $\varepsilon$-Borel set $\ccalB_\varepsilon(\bbS_{k})$ is a continuous set of shift operators $\widetilde{\bbS}_k$ and has countless points [cf. \eqref{eq:BorelSet}]. For each sequence of the shift operators $\bbS_{P:1}\!=\! \{ \bbS_1,\ldots,\bbS_P \} \!\in\! \ccalM$, we can construct the corresponding sequence of the $\varepsilon$-Borel sets $\ccalB_\varepsilon(\bbS_{P:1}) = \{ \ccalB_\varepsilon(\bbS_{1}),\ldots,\ccalB_\varepsilon(\bbS_{P}) \}$. The latter is a set of shift operator sequences $\widetilde{\bbS}_{P:1}\!=\! \{ \widetilde{\bbS}_1,\ldots,\widetilde{\bbS}_P \}$ with each $\widetilde{\bbS}_k \in \ccalB_\varepsilon(\bbS_{k})$ for $k=1,...,P$ and contains also countless points. Given two discrete sequences $\bbS_{P:1}^{(i)}, \bbS_{P:1}^{(j)} \in \ccalM$, the union of the respective $\varepsilon$-Borel set sequences $\big\{\ccalB_\varepsilon\big(\bbS_1^{(i)}\big),\ldots,\ccalB_\varepsilon\big(\bbS_P^{(i)}\big)\big\}$ and $\big\{\ccalB_\varepsilon\big(\bbS_1^{(j)}\big),\ldots,\ccalB_\varepsilon\big(\bbS_P^{(j)}\big)\big\}$ is defined as
\begin{align}\label{eq:unionOperation}
	&\bigcup_{\bbS_{P:1}\in \big\{\bbS_{P:1}^{(i)}, \bbS_{P:1}^{(j)}\big\}} \Big\{\ccalB_\varepsilon\big(\bbS_1\big),...,\ccalB_\varepsilon\big(\bbS_P\big)\Big\} \\ 
	&:= \Big\{\ccalB_\varepsilon\big(\bbS_1^{(i)}\big)\cup\ccalB_\varepsilon\big(\bbS_1^{(j)}\big),\ldots,\ccalB_\varepsilon\big(\bbS_P^{(i)}\big)\cup\ccalB_\varepsilon\big(\bbS_P^{(j)}\big)\Big\} \nonumber
\end{align}
which is also a set of shift operator sequences $\widetilde{\bbS}_{P:1} = \{\widetilde{\bbS}_1, \ldots,\widetilde{\bbS}_P\}$ with each $\widetilde{\bbS}_k \in \ccalB_\varepsilon\big(\bbS_k^{(i)}\big)\cup\ccalB_\varepsilon\big(\bbS_k^{(j)}\big)$ for $k=1,...,P$. This union contains all possible shift operator sequences that belong to the constituted $\varepsilon$-Borel set sequences. We then define the $\varepsilon$-Borel generalization $\widetilde{\ccalM}$ as follows.
\begin{definition}[$\varepsilon$-Borel generalization]\label{def1:BorelSet}
	The $\varepsilon$-Borel generalization of the discrete set $\ccalM$ with shift operator sequences $\bbS_{P:1}$ is defined as the union of the $\varepsilon$-Borel set sequences 
	\begin{align}\label{def:BorelSet}
		\widetilde{\ccalM} := \bigcup_{\bbS_{P:1}\in\ccalM}\Big\{ \ccalB_\varepsilon(\bbS_{1}),\ldots,\ccalB_\varepsilon(\bbS_{P}) \Big\}.
	\end{align}
    where $\bigcup_{\bbS_{P:1}\in\ccalM}$ stands for the union of all $\varepsilon$-Borel set sequences w.r.t. all sequences $\bbS_{P:1} \in \ccalM$ [cf. \eqref{eq:unionOperation}].
\end{definition}
The $\varepsilon$-Borel generalization $\widetilde{\ccalM}$ contains countless points $\widetilde{\bbS}_{P:1} = \{\widetilde{\bbS}_1, \ldots,\widetilde{\bbS}_P\}$ and the expectation of any function $\widetilde{f}(\bbx; \widetilde{\bbS}_{P:1})$ over $\widetilde{\ccalM}$ is 
\begin{align}
	\mathbb{E}[\bbPhi(\bbx;\widetilde{\bbS}_{P:1}, \ccalH)] = \int_{\widetilde{\bbS}_{P:1}\in\widetilde{\ccalM}}\widetilde{f}(\bbx;\widetilde{\bbS}_{P:1}) d \mu(\widetilde{\bbS}_{P:1}) 
\end{align}
where $\mu(\cdot)$ is the probability measure over $\widetilde{\ccalM}$. Such a probability measure is non-atomic, i.e., for any set $\ccalA \in \widetilde{\ccalM}$ with positive measure $\mu(\ccalA) > 0$, there always exists a subset $\ccalA' \subset \ccalA$ such that $0 < \mu(\ccalA') < \mu(\ccalA)$. Given the function $f(\bbx;\bbS_{P:1})$ and the $\varepsilon$-Borel generalization $\widetilde{\ccalM}$, problem \eqref{eq:alternativeVarianceConstrainedProblem} can be seen as a particular instance of
\begin{alignat}{3} \label{eq:generalizedVarianceConstrainedProblem}
	\widetilde{\mathbb{P}}:=   &\min_{\widetilde{f}} \mathbb{E}_{\widetilde{\ccalM}} \big[ \mathbb{E}_\ccalT \big[ \ccalC(\bby, \widetilde{f}(\bbx;\widetilde{\bbS}_{P:1}))\big] \big]             \\
	&  \st \frac{1}{n}\mbE_{\widetilde{\ccalM}} \Big[ \sum_{i=1}^n [\widetilde{f}(\bbx;\widetilde{\bbS}_{P:1})]_i \Big] \ge C_f, \nonumber \\
	&  ~~~~~\frac{1}{n}\mbE_{\widetilde{\ccalM}} \Big[ \sum_{i=1}^n [\widetilde{f}(\bbx;\widetilde{\bbS}_{P:1})]^2_i \Big] \le C_s \nonumber
\end{alignat}
where $\widetilde{f}(\bbx;\widetilde{\bbS}_{P:1})$ is the function defined on $\widetilde{\ccalM}$ and $\widetilde{\bbS}_{P:1}$ is a sequence of random shift operators in $\widetilde{\ccalM}$. We now establish the strong duality for problem \eqref{eq:generalizedVarianceConstrainedProblem}.
\begin{proposition}\label{prop:strongDualityGeneralized}
	Suppose there exists a feasible solution $\widetilde{f}(\bbx;\widetilde{\bbS}_{P:1})$ satisfying the constraints in \eqref{eq:generalizedVarianceConstrainedProblem} with strict inequality. Then, problem \eqref{eq:generalizedVarianceConstrainedProblem} has a null duality gap $\widetilde{\mathbb{P}} = \widetilde{\mathbb{D}}$.
\end{proposition}
\begin{proof}
Define $\bbz_1 = \mbE_{\widetilde{\ccalM}} \big[\widetilde{f}(\bbx;\widetilde{\bbS}_{P:1})\big]$, $\bbz_2 = \mbE_{\widetilde{\ccalM}} \big[\widetilde{f}(\bbx;\widetilde{\bbS}_{P:1})^2\big]$ where $(\cdot)^2$ is the pointwise square operation, and $\bbz = \big[\bbz_1^\top, \bbz_2^\top\big]^\top$. Let $g_1(\bbz) = \sum_{i=1}^n [\bbz_1]_i$, $g_2(\bbz)=\sum_{i=1}^n [\bbz_2]_i$ be functions of $\bbz$. Substituting these representations into problem \eqref{eq:generalizedVarianceConstrainedProblem} yields
\begin{alignat}{3} \label{eq:proofProp12}
	\widetilde{\mathbb{P}}:=   &\min_{\widetilde{f}} \mathbb{E}_{\widetilde{\ccalM}} \left[ \mathbb{E}_\ccalT \left[ \ccalC(\bby, \widetilde{f}(\bbx;\widetilde{\bbS}_{P:1}))\right] \right],             \\
	&  \st~ -g_1(\bbz) + C_f \le 0,~g_2(\bbz) - C_s \le 0, \nonumber \\
	&  ~~~~~~ \bbz = \big[\mbE_{\widetilde{\ccalM}} \big[\widetilde{f}(\bbx;\widetilde{\bbS}_{P:1})\big]^\top, \mbE_{\widetilde{\ccalM}} \big[\widetilde{f}(\bbx;\widetilde{\bbS}_{P:1})^2\big]^\top\big]. \nonumber
\end{alignat}
Since $-g_1(\bbz)$ and $g_2(\bbz)$ are convex functions of $\bbz$, problem \eqref{eq:proofProp12} can be considered as a sparse functional program \cite{chamon2020functional}. By using \cite[Theorem 1]{chamon2020functional}, we prove the strong duality $\widetilde{\mathbb{P}} = \widetilde{\mathbb{D}}$. Note that $-g_1(\bbz)$ and $g_2(\bbz)$ are also composite functions of ($\bbx$, $\widetilde{\bbS}_{P:1}$), which integrally may not be convex. But from the condition in \cite{chamon2020functional}, we need only the outer form convex but not the composite form.  
\end{proof}
\noindent That is, problem \eqref{eq:generalizedVarianceConstrainedProblem} can be solved in the dual domain without loss of optimality. We leverage this results to characterize the duality gap of problem \eqref{eq:alternativeVarianceConstrainedProblem} where the SGNN $\bbPhi(\bbx;\bbS_{P:1}, \ccalH)$ operates over a discrete set $\ccalM$.

\begin{remark}\label{remark1}
	 Proposition \ref{prop:strongDualityGeneralized} proves the null duality gap for the general version of the surrogate problem \eqref{eq:alternativeVarianceConstrainedProblem}. If we were to consider the general version of the original problem \eqref{eq:varianceConstrainedProblem}, we would have not proven such strong duality. This is because the variance constraint in problem \eqref{eq:varianceConstrainedProblem} takes the form \vspace{-2mm}
	 \begin{align}
	 	g(\bbz) - C_v \le 0~\text{with}~ g(\bbz) = \sum_{i=1}^n [\bbz_2]_i - [\bbz_1]_i^2.
	 \end{align}
     Since $g(\bbz)$ is a non-convex function of $\bbz$, the conditions of \cite[Theorem 1]{chamon2020functional} do not apply. 
\end{remark} 

\subsection{Duality Analysis}\label{subsec:nearUniversal}

\begin{figure*}%
	\centering
	\begin{subfigure}{0.65\columnwidth}
		\includegraphics[width=1.0\linewidth, height = 0.75\linewidth]{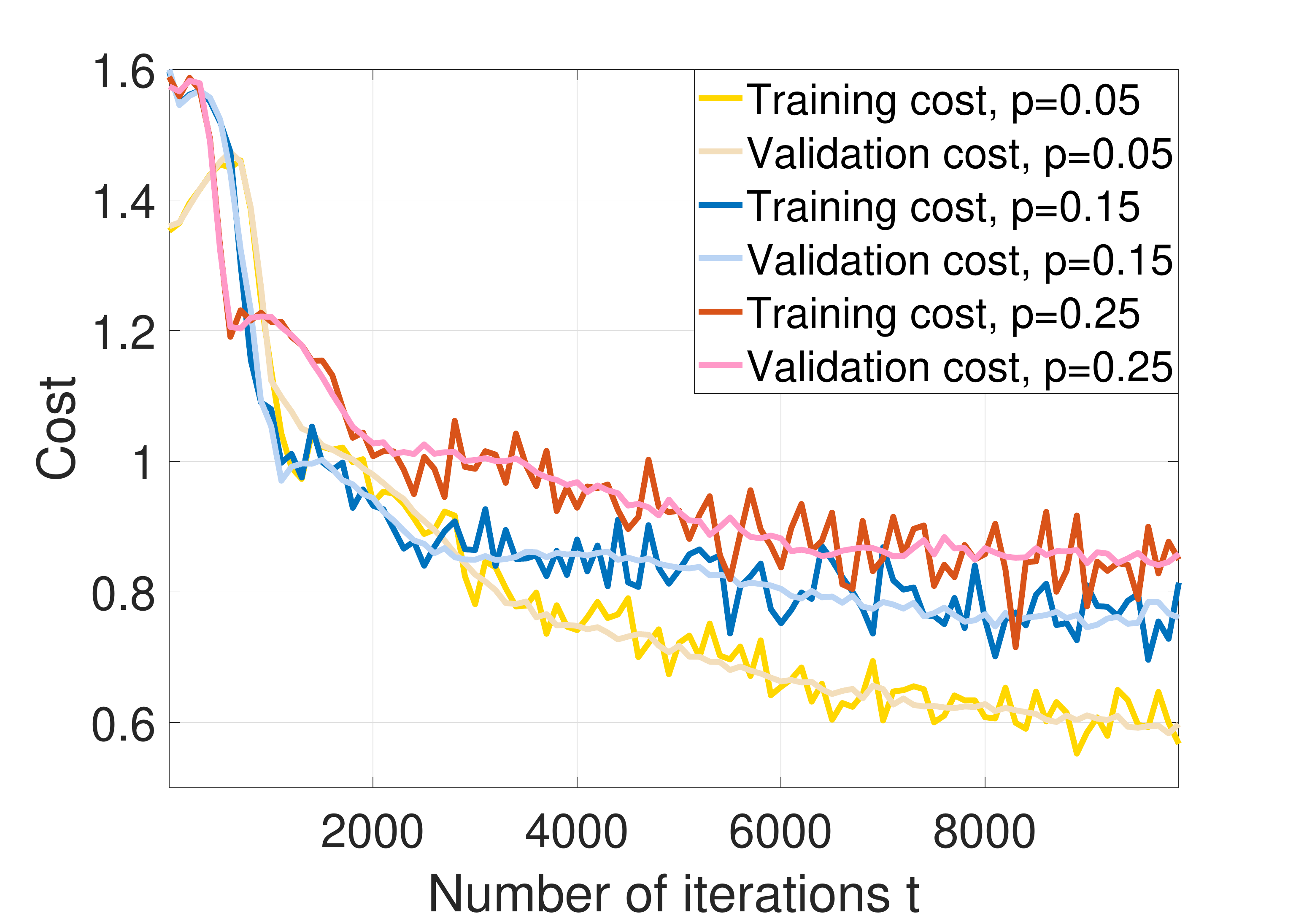}%
		\caption{}%
		\label{fig0}%
	\end{subfigure}\hfill\hfill
	\begin{subfigure}{0.65\columnwidth}
		\includegraphics[width=1.0\linewidth,height = 0.75\linewidth]{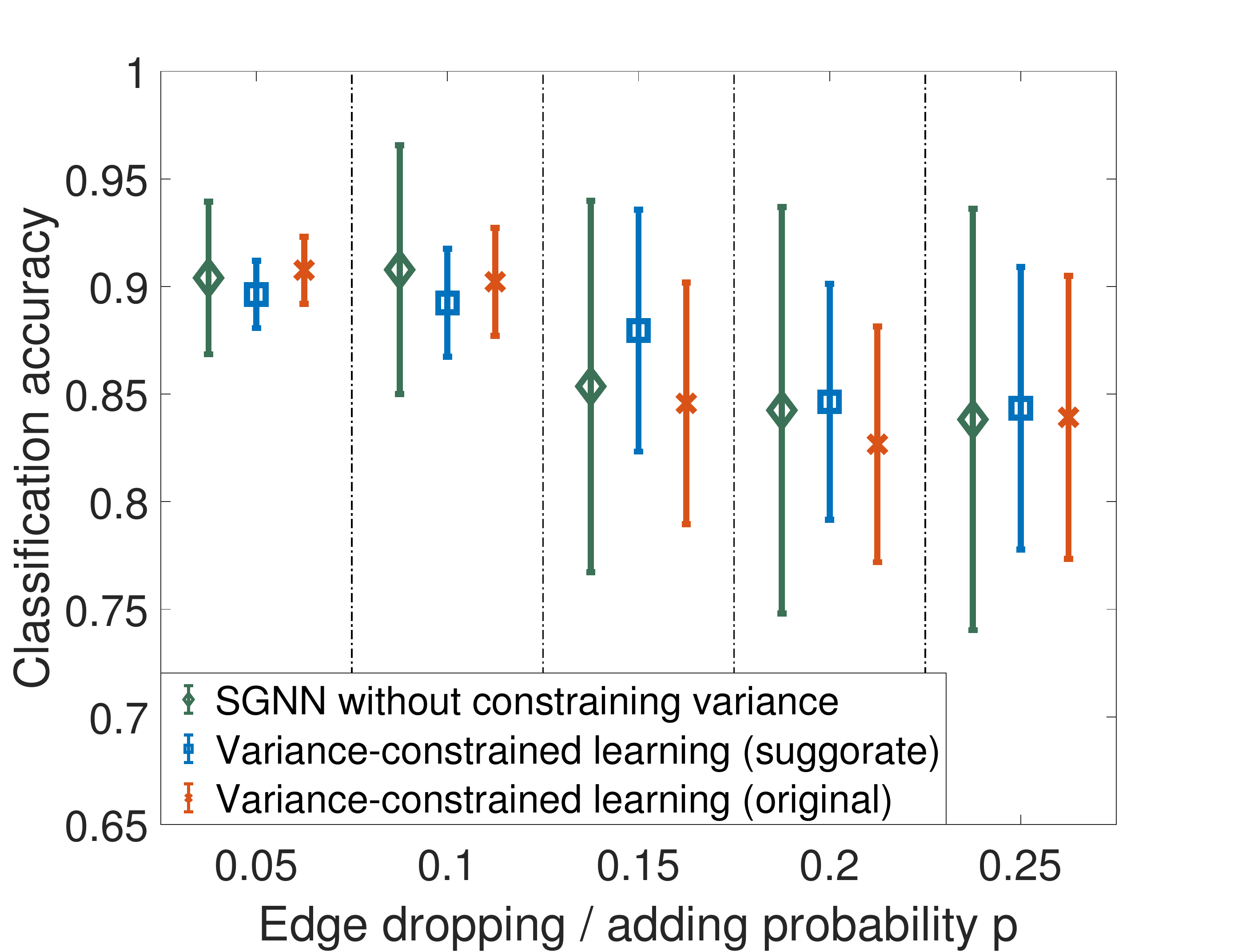}%
		\caption{}%
		\label{fig2}%
	\end{subfigure}\hfill\hfill
	\begin{subfigure}{0.65\columnwidth}
		\includegraphics[width=1.0\linewidth,height = 0.75\linewidth]{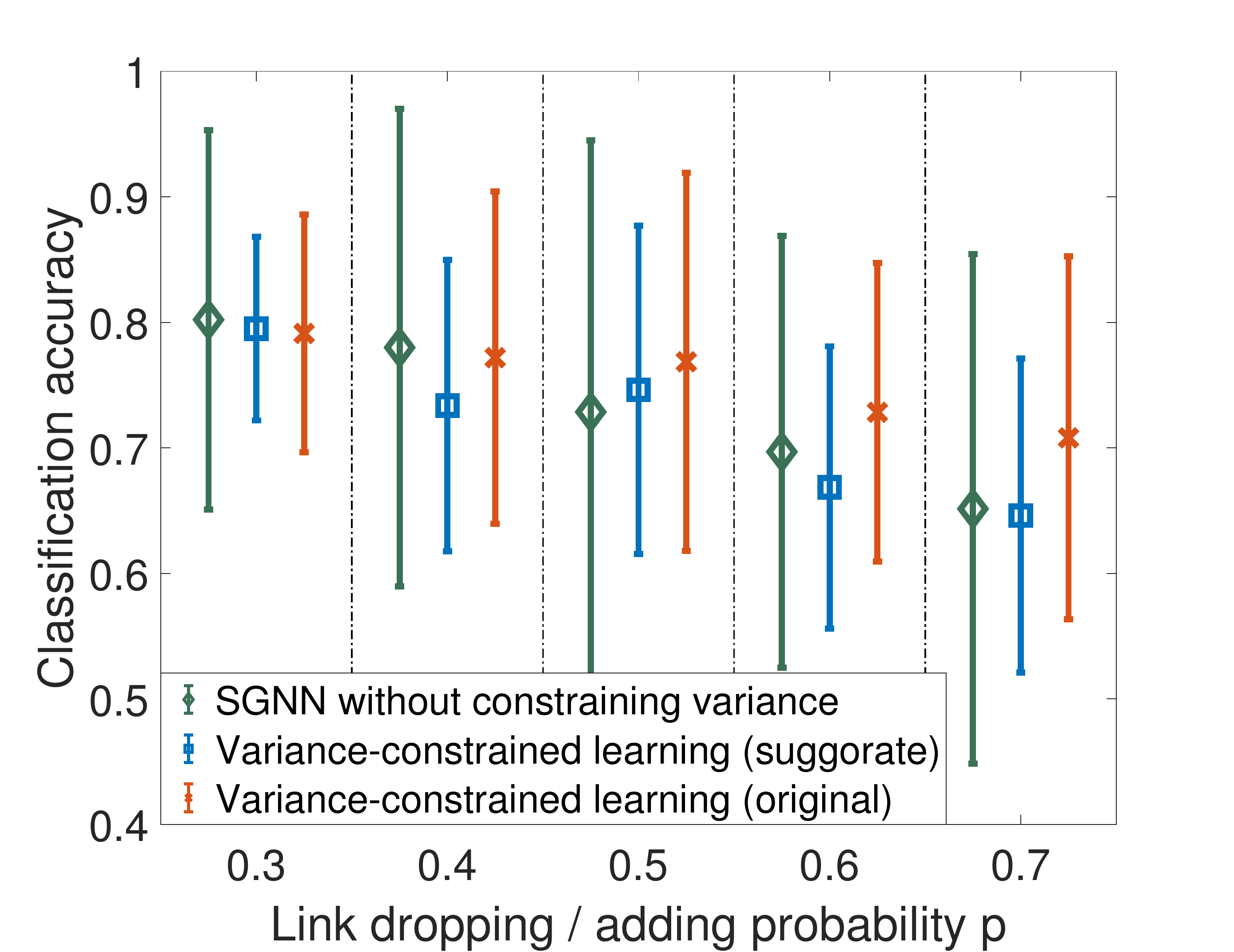}%
		\caption{}%
		\label{fig333}%
	\end{subfigure}
	\caption{(a) Convergence of the cost with different edge dropping probabilities. (b) Expected classification accuracy and standard deviation with and without (w/o) the variance-constrained learning for source localization. (c) Performance with large edge dropping probabilities $p$.
	}\label{fig:vary_n1}
\end{figure*}

We now analyze the duality deviation induced by the problem generalization. First, we particularize the function $\widetilde{f}(\bbx;\widetilde{\bbS}_{P:1})$ to the SGNN $\bbPhi(\bbx;\widetilde{\bbS}_{P:1}, \ccalH)$ via the 
\emph{$\epsilon$-universal parameterization}. 
\begin{definition}[$\eps$-universal parameterization]\label{def:universalParameterization}
A parameterization\footnote{A parameterization is defined as a mathematical model that represents some mapping as a function of some independent parameters} $\bbPhi(\bbx;\widetilde{\bbS}_{P:1}, \ccalH)$ is $\epsilon$-universal if for any function $\widetilde{f}(\bbx;\widetilde{\bbS}_{P:1})$ in the considered domain, there exist a set of parameters $\ccalH$ such that 
\begin{align}
	\mathbb{E}_{\widetilde{\ccalM}}\big[\| \widetilde{f}(\bbx;\widetilde{\bbS}_{P:1}) - \bbPhi(\bbx;\widetilde{\bbS}_{P:1}, \ccalH) \|^2\big] \le \epsilon^2
\end{align}
where the expectation $\mathbb{E}_{\widetilde{\ccalM}}[\cdot]$ is over the generalized set $\widetilde{\ccalM}$ of the shift operator sequence $\widetilde{\bbS}_{P:1}$.
\end{definition} 
\noindent 
An $\epsilon$-universal parametrization can model any function in the considered domain within some accuracy $\epsilon$. Such a property has been shown in a number of machine learning architectures, including radial basis function networks \cite{park1991universal}, reproducing kernel Hilbert spaces \cite{sriperumbudur2010relation}, and deep neural networks \cite{hornik1991approximation}.
\begin{assumption}\label{as:universalParameterization}
	For a given 
	SGNN $\bbPhi(\bbx;\widetilde{\bbS}_{P:1}, \ccalH)$, there exists a finite accuracy $\epsilon>0$ such that the SGNN is an $\epsilon$-universal parameterization w.r.t. the generalized set $\widetilde{\ccalM}$.
\end{assumption} 
\noindent Assumption \ref{as:universalParameterization} 
implies that for the considered SGNN, there exists some finite $\epsilon > 0$ to make it an $\eps$-universal parameterization. The value of $\eps$ depends on the representational capacity of the considered SGNN, i.e., a deeper (high $L$) and wider (high $F$) SGNN may have a higher representational capacity and we may choose a smaller $\eps$ for it w.r.t. a stronger $\eps$-universal parameterization. This property characterizes the deviation induced by particularizing $\widetilde{f}(\bbx;\widetilde{\bbS}_{P:1})$ to $\bbPhi(\bbx;\widetilde{\bbS}_{P:1}, \ccalH)$ and will be reflected in the duality gap. Second, we particularize the continuous set $\widetilde{\ccalM}$ to the discrete set $\ccalM$. The relation between these two sets is characterized by the $\varepsilon$-Borel set [Def. \ref{def1:BorelSet0}]. To proceed the analysis, we assume the following. 
\begin{assumption}\label{as:LipschitzCost}
	The loss $\ccalC(\cdot, \cdot)$ is Lipschitz over $\ccalT = \{(\bbx,\bby)\}$, i.e., for any $\bby_1$ and $\bby_2$, there exists a constant $C_\ell$ such that 
	\begin{align}
		&\big| \mathbb{E}_\ccalT [\ccalC(\bby, \bby_1)] \!-\! \mathbb{E}_\ccalT [\ccalC(\bby, \bby_2)] \big| \!\le\! C_\ell \| \bby_1 \!-\! \bby_2 \|. 
	\end{align} 
\end{assumption} 
\begin{assumption}\label{as:inputBound}
	The SGNN output $\bbPhi(\bbx;\bbS_{P:1}, \ccalH)$ is bounded, i.e., there exists a constant $C_y$ s.t. $\|\bbPhi(\bbx;\bbS_{P:1}, \ccalH)\| \le C_y$.
\end{assumption} 
\noindent Assumption \ref{as:LipschitzCost} is a continuity statement on the loss $\ccalC(\cdot, \cdot)$, which is common in optimization theory \cite{boyd2004convex} and holds for popular classification and regression losses. Assumption \ref{as:inputBound} considers the SGNN output bounded by a constant $C_y$ independent of the filter coefficients, which has been proven 
for the SGNN in Lemma \ref{lemma:outputBound} of Appendix \ref{lemmasProof}.

The following theorem shows the duality gap of problem \eqref{eq:alternativeVarianceConstrainedProblem}.  
\begin{theorem}\label{thm:dualityGap}
	Consider problem \eqref{eq:alternativeVarianceConstrainedProblem} with primal and dual solutions $\mathbb{P}$ and $\mathbb{D}$, respectively. Let the SGNN $\bbPhi(\bbx;\bbS_{P:1}, \ccalH)$ be of $L$ layers comprising $F$ filters of order $K$. Let the frequency response \eqref{eq:generalizeFrequencyResponse} of these filters satisfy Assumption \ref{as:LipschitzFilter} with $C_L$ and the nonlinearity $\sigma(\cdot)$ satisfy Assumptions \ref{as:LipschitzNonlinearity1}-\ref{as:LipschitzNonlinearity2} with $C_\sigma$. Let also the SGNN satisfy Assumption \ref{as:universalParameterization} w.r.t. the $\varepsilon$-Borel generalization $\widetilde{\ccalM}$ with $\epsilon$, its output be bounded according to Assumption \ref{as:inputBound} with $C_y$, and the cost function $\ccalC(\cdot, \cdot)$ satisfy Assumption \ref{as:LipschitzCost} with $C_\ell$. Then, the duality gap of problem \eqref{eq:alternativeVarianceConstrainedProblem} is bounded by
\begin{align}\label{eq:mainThmDualityGap}
	|\mathbb{P}\!-\!\mathbb{D}| \!\le & \Big( C_\ell + \frac{\widetilde{\gamma}_1^*}{\sqrt{n}} \!+\! \widetilde{\gamma}_2^*\big(\frac{2C_y}{\sqrt{n}} \!+\! \frac{\eps}{n}\big)\! \Big) \eps \!+\!	C \varepsilon \!+\! \ccalO(\varepsilon^2)
\end{align}
where $\widetilde{\bbgamma}^* = [\widetilde{\gamma}_1^*,\widetilde{\gamma}_2^*]^\top$ is the optimal dual variable of problem \eqref{eq:generalizedVarianceConstrainedProblem} and $C$ is a constant related to the SGNN architectural properties -- see the expression of $C$ in \eqref{proof:prop3eq20}.
\end{theorem} 
\begin{proof}
	See Appendix \ref{pr:proofTheorem3}.
\end{proof} 
\noindent The result indicates that the duality gap is induced by two types of errors: the parameterization error $\epsilon$ of the SGNN [Def. \ref{def:universalParameterization}] and the generalization error $\varepsilon$ of the set [Def. \ref{def1:BorelSet}]. The parameterization error is present in the first term of \eqref{eq:mainThmDualityGap}, which becomes small when the SGNN exhibits a strong representational capacity to approximate unparameterized functions. This is an irreducible error that tells how well the SGNN covers the function space and exists for any GNN solutions. The generalization error is present in the second term of \eqref{eq:mainThmDualityGap}, which can be sufficiently small by considering small Borel sets that satisfy Assumption \ref{as:universalParameterization}. A small duality gap indicates that solving the problem in the dual domain comes with a contained optimality loss, compared to solving it directly in the primal domain,
which justifies the primal-dual learning procedure.

Theorem \ref{thm:dualityGap} discusses the duality gap induced by solving problem \eqref{eq:alternativeVarianceConstrainedProblem} in the dual domain exactly. However, \emph{it is still unclear if the primal-dual learning procedure [Alg.~\ref{alg:primalDualAlgorithm}] even converges to a neighboorhood of the dual solution $\mathbb{D}$.} In the next section, we answer this question affirmative 
and combine the convergence error with the duality gap to provide a unified performance analysis. 

\section{Convergence}\label{sec:Convergence}



\begin{figure*}%
	\centering
	\begin{subfigure}{0.485\columnwidth}
		\includegraphics[width=1\linewidth, height = 0.75\linewidth]{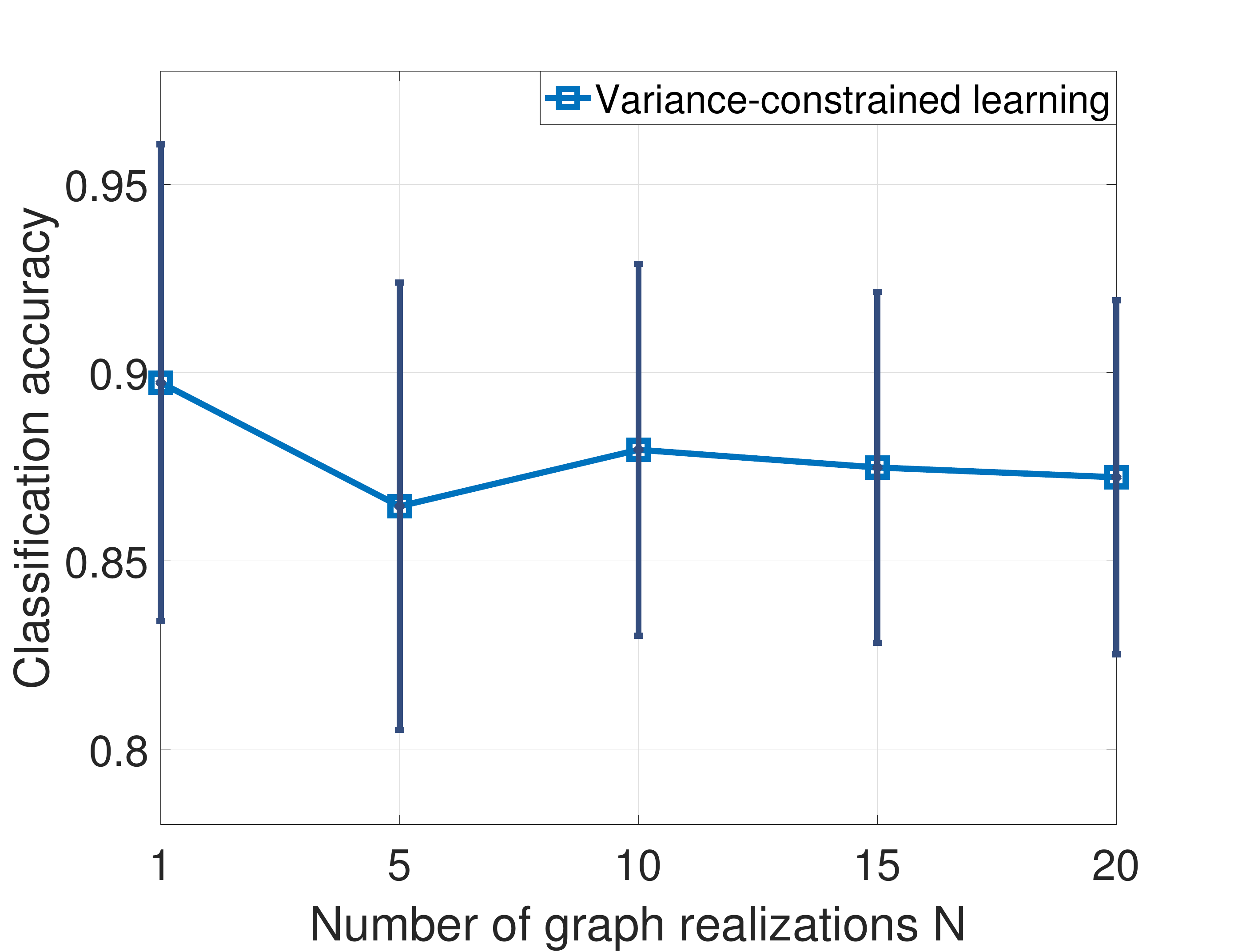}%
		\caption{}%
		\label{fig_different_samples}%
	\end{subfigure}\hfill\hfill%
	\begin{subfigure}{0.485\columnwidth}
		\includegraphics[width=1\linewidth, height = 0.75\linewidth]{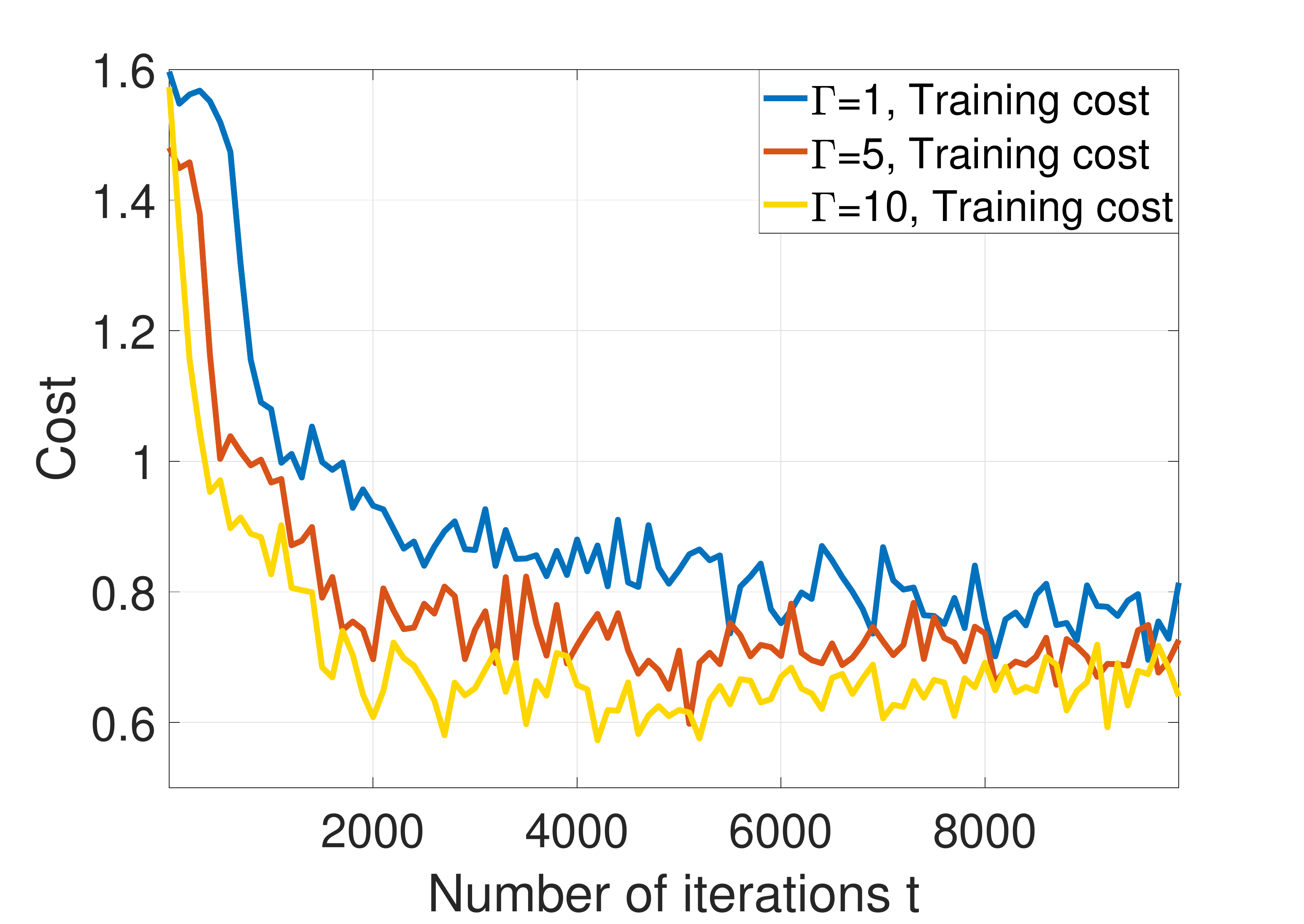}%
		\caption{}%
		\label{fig_different_steps}%
	\end{subfigure}\hfill\hfill%
	\begin{subfigure}{0.485\columnwidth}
		\includegraphics[width=1\linewidth,height = 0.75\linewidth]{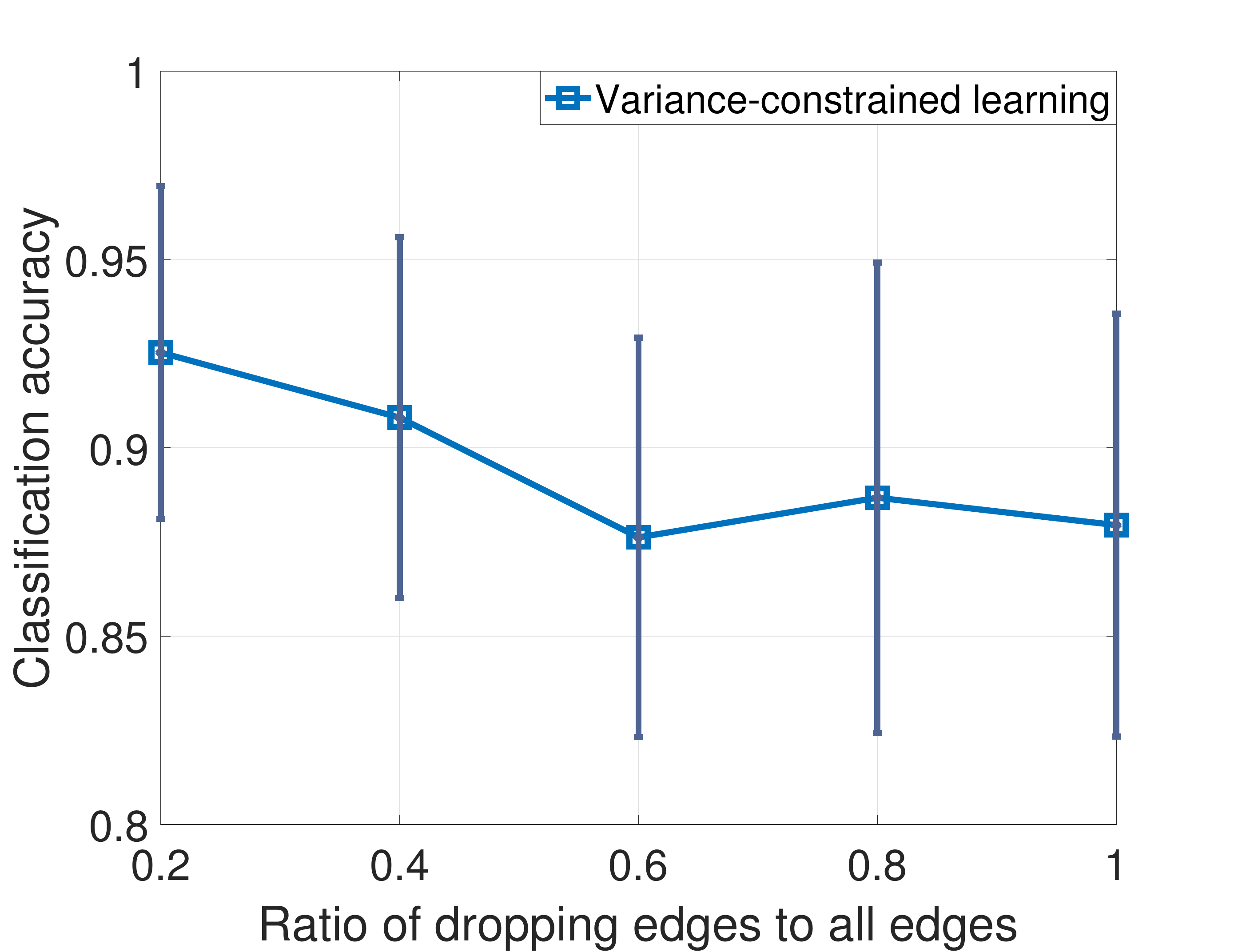}%
		\caption{}%
		\label{fig_different_edges}%
	\end{subfigure}\hfill\hfill%
	\begin{subfigure}{0.485\columnwidth}
		\includegraphics[width=1\linewidth,height = 0.75\linewidth]{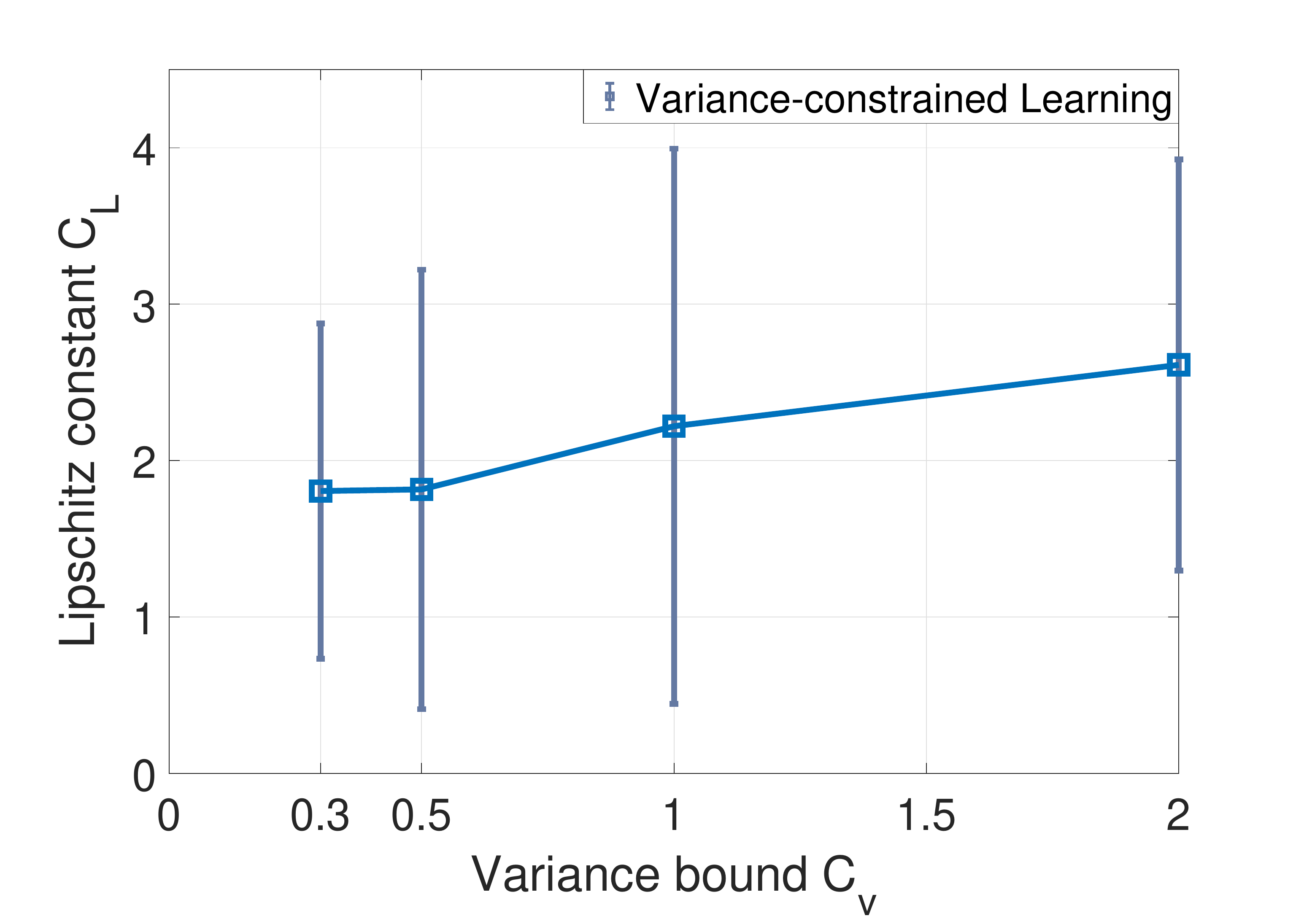}%
		\caption{}%
		\label{fig_Lipschitz}%
	\end{subfigure}
	\caption{(a) Expected classification accuracy and standard deviation with different numbers of GRES($p,q$) realizations $N$. (b) Convergence of the cost with different numbers of gradient steps $\Gamma$ at the primal phase. (c) Expected classification accuracy and standard deviation with different numbers of dropping edges $M_d$. (d) Expected Lipschitz constant $C_L$ with different variance bounds $C_v$. }\label{fig:vary_hyparameters}
\end{figure*}

The main challenge to prove the convergence of the primal-dual learning, stands in the facts that we approximate the minimization at the primal phase with stochastic gradient descent [cf. \eqref{eq:dualFunction}] and every $\Gamma$ primal updates we run a single dual update. To characterize this convergence, we make the following mild assumption. 
\begin{assumption}\label{as:primalOptimalityLoss}
 Let $\ccalH^*$ be the minima of the Lagrangian $\mathcal{L}(\ccalH,\bbgamma)$ [cf. \eqref{eq:Lagrangian}] and $\ccalH^{(\Gamma)}$ the approximate solution obtained by the primal phase with gradient descent [cf. \eqref{eq_priup}]. There exists a constant $\xi \ge 0$ such that for any dual variable $\bbgamma \in \mathbb{R}_+$, it holds that 
	\begin{align}\label{eq:primalOptimalLoss}
		| \mathcal{L}(\ccalH^*,\bbgamma) - \mathcal{L}(\ccalH^{(\Gamma)},\bbgamma) | \le \xi.
	\end{align} 
\end{assumption} 
\noindent That is, the gradient descent applied at the primal phase solves the dual function $\mathcal{D}(\bbgamma) = \min_{\ccalH} \mathcal{L}(\ccalH,\bbgamma)$ within an error neighborhood $\xi$. The value of $\xi$ depends on the performance of the gradient descent, which has exhibited success in a wide array of optimization problems \cite{bottou2012stochastic}. The following theorem then establishes the convergence result. 

\begin{theorem}\label{thm:convergencePrimalDual}
	Consider the primal-dual learning for problem \eqref{eq:alternativeVarianceConstrainedProblem} [Alg. \ref{alg:primalDualAlgorithm}]. Let the SGNN output satisfy Assumption \ref{as:inputBound} with $C_y$ and the primal phase satisfy Assumption \ref{as:primalOptimalityLoss} with $\xi$. Then, for any desirable accuracy $\delta > 0$, Algorithm \ref{alg:primalDualAlgorithm} converges to an error neighborhood of the dual solution $\mathbb{D}$ of problem \eqref{eq:alternativeVarianceConstrainedProblem} as 
\begin{align}\label{eq:convergenceResult}
	&|\ccalL(\ccalH^{(\Gamma)}_T\!, \bbgamma_T)\!-\! \mathbb{D}| \!\le\! 2\xi \!+\! \frac{\Big( \!\big(C_f \!+\! \frac{C_y}{\sqrt{n}}\big)^2 \!+\! \big(C_s \!+\! \frac{C_y^2}{n}\big)^2 \Big)}{2}\eta_\gamma \!+ \delta
\end{align}
    in at most $T$ iterations with $T \le \|\bbgamma_0 - \bbgamma^*\|^2/(2 \eta_\gamma \delta)$, where $\bbgamma_0$ and $\bbgamma^*$ are the initial and optimal dual variables for the dual problem [cf. \eqref{eq:dualFunction}], and $\eta_\gamma$ is the dual step-size.
\end{theorem} 

\begin{proof}
  See Appendix \ref{pr:theorem4}.
\end{proof} 

\noindent Theorem \ref{thm:convergencePrimalDual} states that the primal-dual learning converges to an error neighborhood of the dual solution within a finite number of iterations that is inversely proportional to the desirable accuracy $\delta$. The error size depends on the suboptimality of the solution of the primal phase and the step-size of the dual phase. Inspecting \eqref{eq:convergenceResult}, the error size consists of three terms:

\begin{enumerate}[(1)]
	
	\item The first term $2\xi$ decreases when we perform sufficient gradient steps at the primal phase and the parameters $\ccalH_t^{(\Gamma)}$ [cf. \eqref{eq_priup}] are close to the optimal $\ccalH_{t}^*$ at iteration $t$.
	
	\item The second term is proportional to the dual step-size $\eta_\gamma$, which could be set sufficiently small [cf. \eqref{eq:dualUpdate}].
	
	\item The third term $\delta$ is inversely proportional to the number of iterations $T$, which decreases if we run the primal-dual learning for more iterations.
	
\end{enumerate}

By combining Theorems \ref{thm:dualityGap}-\ref{thm:convergencePrimalDual}, we can characterize completely the solution suboptimality of the primal-dual learning procedure w.r.t. both the duality gap and the iterative method. 
\begin{corollary}\label{coro:optimallossPrimalDual}
	Under the same settings of Theorems \ref{thm:dualityGap}-\ref{thm:convergencePrimalDual}, the suboptimality of solving problem \eqref{eq:alternativeVarianceConstrainedProblem} with the primal-dual learning procedure can be bounded as 
	\begin{align}\label{eq:optimalLoss}
		|\ccalL(\ccalH_T^{(\Gamma)}\!,\! \bbgamma_T) \!-\! \mathbb{P}| &\le C_1 \eps \!+\! 2 \xi \!+\! C_2 \eta_\gamma \!+\! \sigma \!+\! C_3 \varepsilon \!+\! \ccalO(\varepsilon^2)
	\end{align} 
    where constants $C_1, C_2$ are specified in \eqref{eq:mainThmDualityGap} and constant $C_3$ is specified in \eqref{eq:convergenceResult}.
\end{corollary} 
\noindent This result indicates that the proposed variance-constrained learning converges to a solution $\ccalL(\ccalH_T^{(\Gamma)}, \bbgamma_T)$ in the dual domain within a finite number of iterations, which is close to the optimal solution $\mathbb{P}$ of the formulated problem \eqref{eq:alternativeVarianceConstrainedProblem}. 

\begin{remark}
	The convergence result \eqref{eq:convergenceResult} holds when the primal phase obtains parameters in a neighborhood of the global solution [cf. \eqref{eq:primalOptimalLoss}]. Since working with neural networks is in a non-convex setting, it is likely to obtain parameters close to a local minima. In this context, \eqref{eq:convergenceResult} indicates what can be achieved at best via the primal-dual learning and what is our handle to control it. We corroborate next that the variance-constrained learning converges satisfactorily in numerical simulations.
\end{remark} 

\section{Numerical Simulations} \label{numer}



We compare the variance-constrained learning with the vanilla GNN and the SGNN using synthetic data from source localization and real data from recommender systems \cite{harper2015movielens}. The vanilla GNN is the standard GNN trained over the deterministic underlying graph \cite{gama2020graphs} and has the same architecture hyper-parameters as the SGNN. In the stochastic setting, the vanilla GNN has shown a lower performance compared with the SGNN \cite{gao2021stochastic} and thus, we focus principally on the comparison with the latter and report the performance of the vanilla GNN as a baseline. For all architectures we tested both the stochastic gradient descent and the ADAM optimizer \cite{Ba2010} for training, while used the latter because it has consistently shown a better performance. The learning rate is $\mu = 10^{-3}$ and decaying factors are $\beta_1=0.9, \beta_2=0.999$. The assumptions made in Sec. \ref{DOL} - \ref{sec:Convergence} typically hold for these practical applications, where the graph signals, graph eigenvalues and architecture parameters are of finite values, while the assumption constants depends on specific problem settings that vary among different applications. Moreover, these are assumed properties for theoretical analysis to shed insights on the proposed algorithm but are not necessary for the algorithm implementation. 

\subsection{Source Localization}

\begin{figure*}%
	\centering
	\begin{subfigure}{0.65\columnwidth}
		\includegraphics[width=1.0\linewidth, height = 0.75\linewidth]{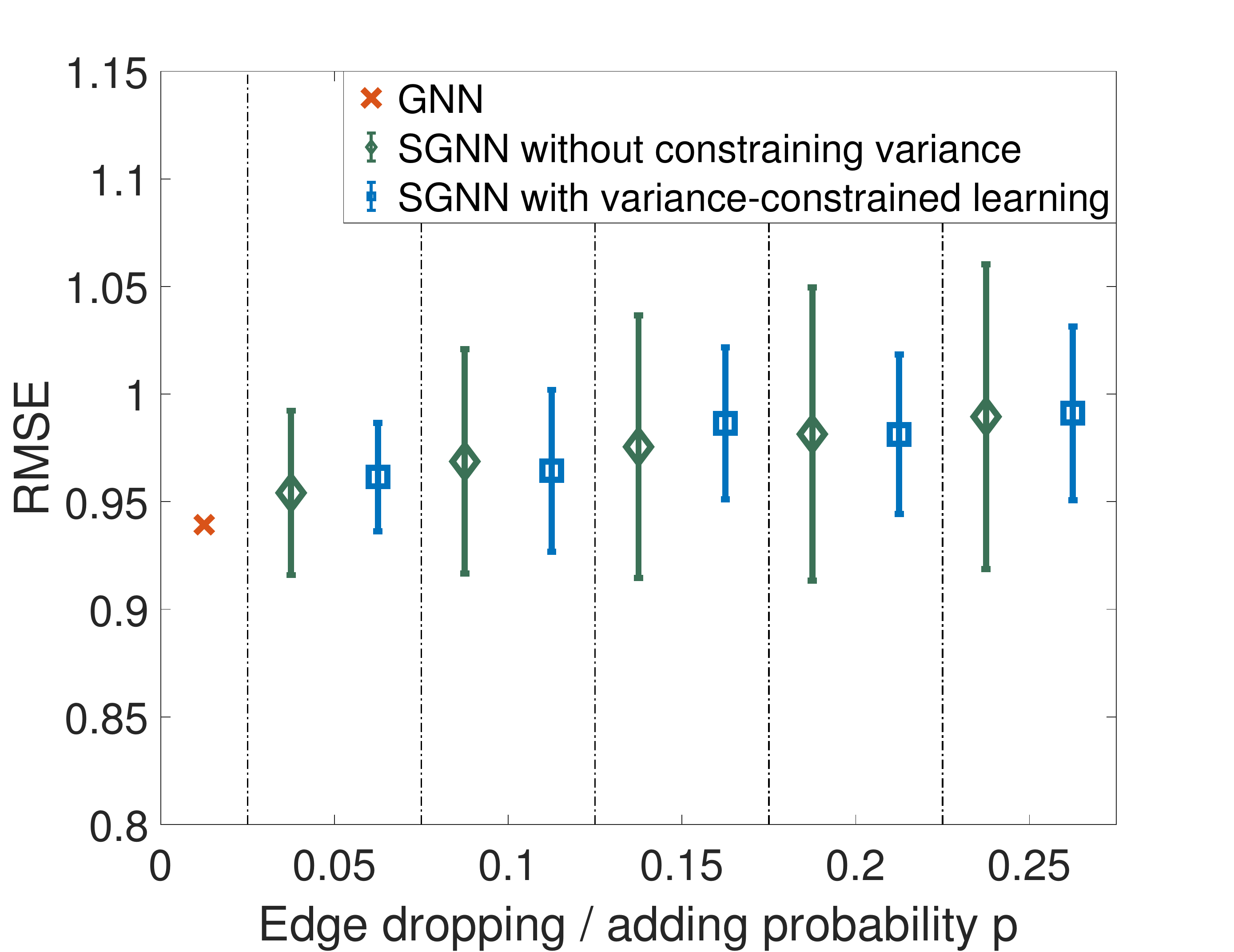}%
		\caption{}%
		\label{subfig:recommendationRMSE}%
	\end{subfigure}\hfill\hfill%
	\begin{subfigure}{0.65\columnwidth}
		\includegraphics[width=1.0\linewidth,height = 0.75\linewidth]{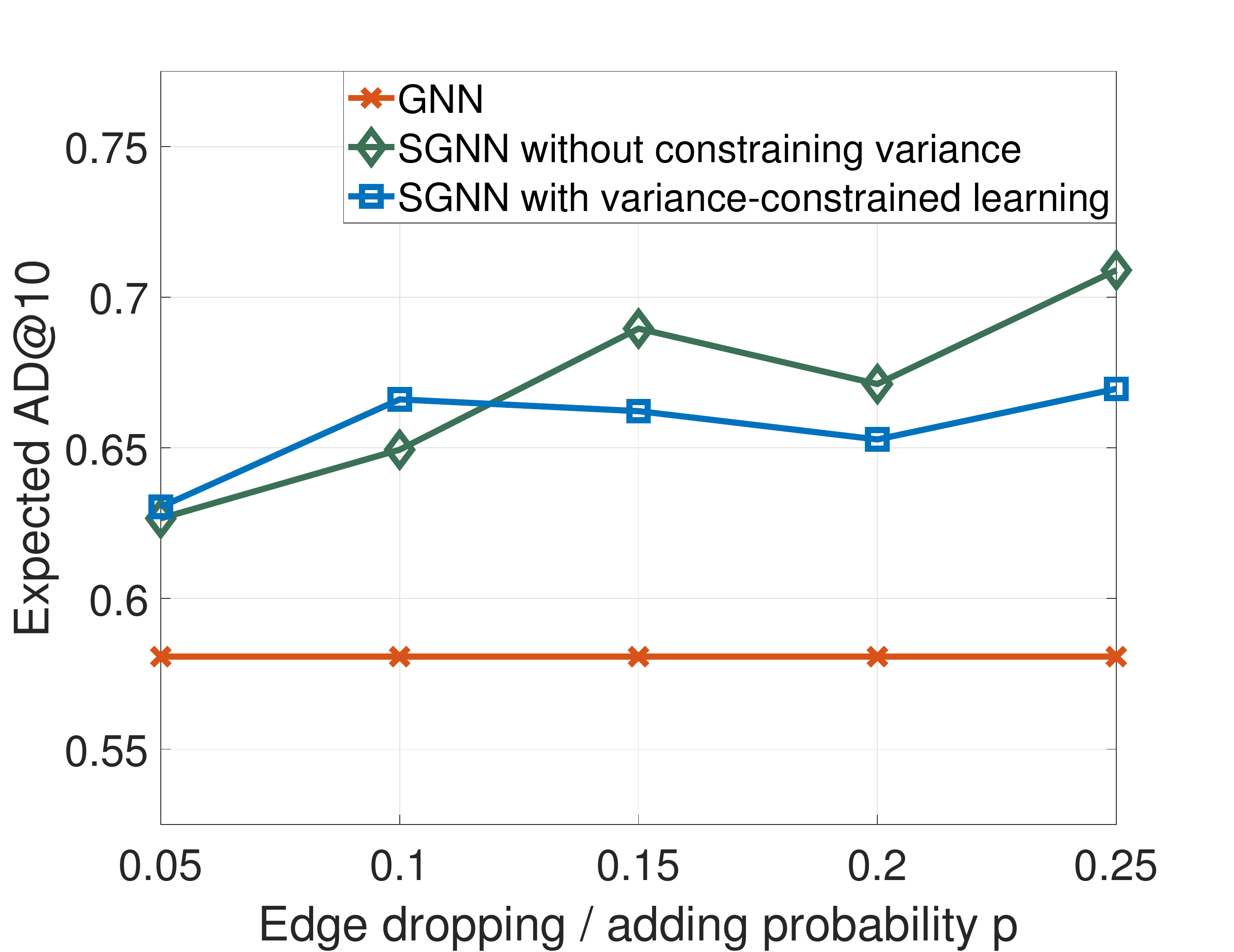}%
		\caption{}%
		\label{subfig:recommendationAD}%
	\end{subfigure}\hfill\hfill%
	\begin{subfigure}{0.65\columnwidth}
		\includegraphics[width=1.0\linewidth,height = 0.75\linewidth]{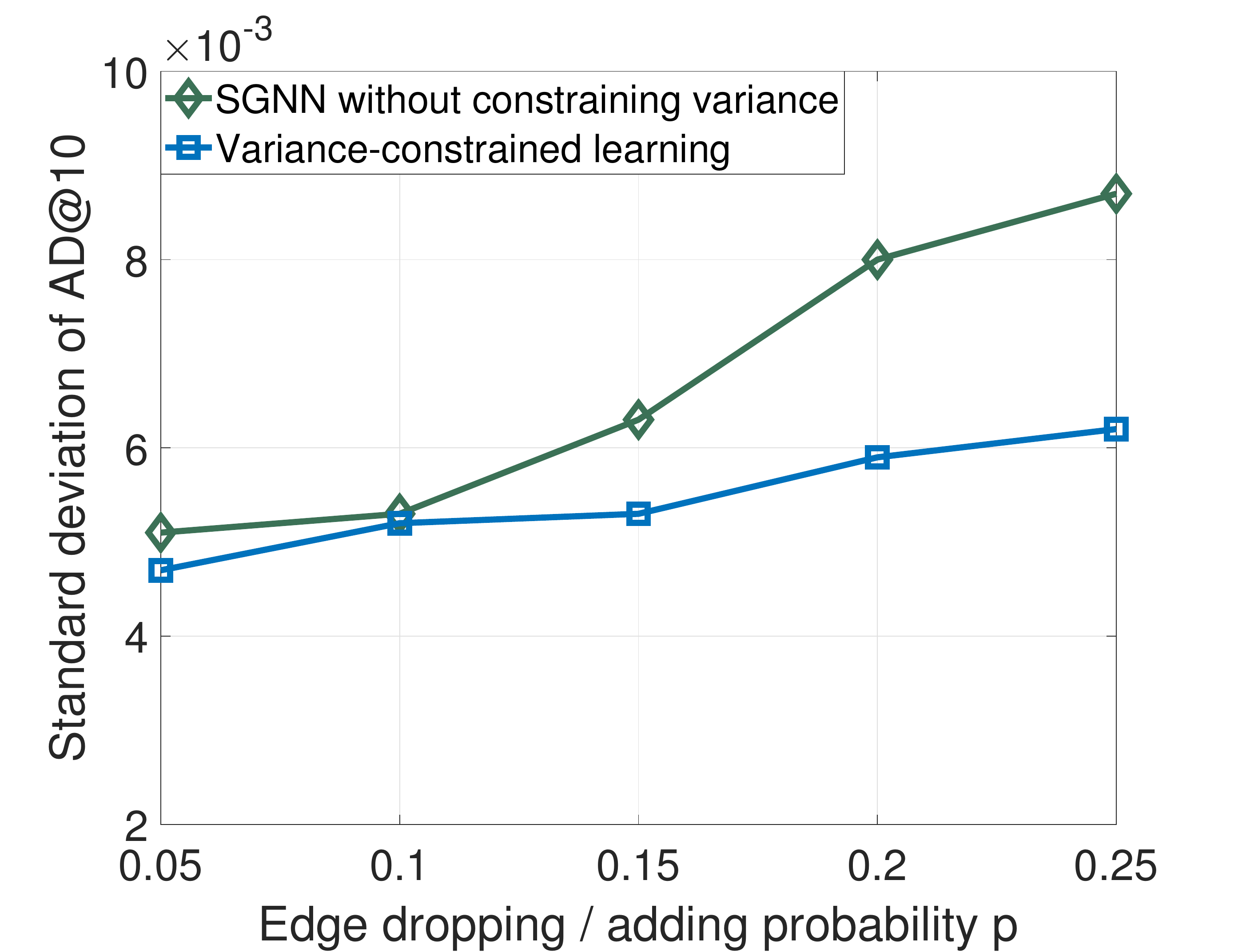}%
		\caption{}%
		\label{subfig:recommendationAD2}%
	\end{subfigure}
	\caption{(a) Expected RMSE and standard deviation of the GNN, the SGNN w/o the variance-constrained learning for movie recommendation. (b)-(c) Expected AD and standard deviation of the GNN, the SGNN w/o the variance-constrained learning for movie recommendation. }\label{fig:recommendation}
\end{figure*}

We consider a diffusion process over a stochastic block model (SBM) graph of $50$ nodes divided into $5$ communities, where the intra- and inter-community edge probabilities are $0.8$ and $0.2$. The goal is to find the community originating the diffusion distributively at a node. The initial graph signal is a Kronecker delta $\bbdelta_{s} \in \reals^{50}$ originated at a source node $s \in \{ s_1,\ldots,s_5 \}$ of a community, where $\{ s_1,\ldots,s_5 \}$ are the five source nodes of five communities respectively. The signal at time $t$ is $\bbx_{s}^{(t)} = \bbS^t \bbdelta_{s} +\bbn$ with $\bbn \in \reals^{50}$ a zero-mean Gaussian noise. We generate $15000$ samples by randomly selecting a source node $s$ and a diffused time $t \in [0,50]$, which are split into $10000$, $2500$, and $2500$ samples for training, validation, and testing, respectively. We consider all edges of the nominal graph may fall with a probability $p$ due to channel fading effects and no edges are added during testing, according to the GRES($p,q$) model with $q=0$. The SGNN has two layers, each with $F=32$ filters of order $K=8$ and the ReLU nonlinearity. The mini-batch contains $50$ samples and the cost function is the cross entropy. The constraint bounds are selected empirically from the confidential interval $[0, 1]$ via validation and are set to $C_f = 0, C_s = 0.5$, i.e., $C_v=0.5$ according to \eqref{eq:resultConstrainedVariance}. The performance is measured by the classification accuracy and the results are averaged over $10$ SBM graph realizations, conditioned on which different graph stochasticity scenarios are investigated.

\smallskip
\noindent \textbf{Performance.} First, we corroborate the convergence of the variance-constrained learning. Fig. \ref{fig0} displays the primal-dual learning procedure over $10000$ iterations with the edge dropping probability $p=0.05$, $0.15$ and $0.25$. The expected cost decreases with the number of iterations, while the decreasing rate reduces gradually; ultimately, approaching a stationary point in all cases. The expected cost of $p = 0.05$ converges slightly later than that of $p=0.15$, $0.25$ because $p = 0.05$ yields a more stable graph with a better performance, such that it takes more iterations to reach a lower cost. The convergent value increases with the edge dropping probability $p$ because of the increased graph randomness.

Then, we compare the performance of the SGNN w/o the variance-constrained learning w.r.t. both the surrogate problem \eqref{eq:alternativeVarianceConstrainedProblem} and the original problem \eqref{eq:varianceConstrainedProblem}. Fig. \ref{fig2} and Fig. \ref{fig333} shows the classification accuracy under different edge dropping probabilities: (mild) $p \in [0.05, 0.25]$ and (harsh) $p \in [0.3, 0.7]$. The variance-constrained learning exhibits a better performance with a comparable expected value and a lower standard deviation. The latter is emphasized when $p$ increases, i.e., when more edges are dropped. The expected performance degrades as $p$ increases, which can be explained by the increased graph variation. The variance-constrained learning maintains a smaller standard deviation, while the unconstrained training increases the standard deviation inevitably. For small probabilities $p$, the variance-constrained learning w.r.t. the surrogate problem \eqref{eq:alternativeVarianceConstrainedProblem} performs comparably to that w.r.t. the original problem \eqref{eq:varianceConstrainedProblem}. We attribute the latter to the fact that the surrogate constraints in \eqref{eq:alternativeVarianceConstrainedProblem} provide similar guarantees on stochastic deviations as the variance constraint in \eqref{eq:varianceConstrainedProblem} [cf. \eqref{eq:resultConstrainedVariance}]. For large probabilities $p$, the surrogate exhibits lower expected performance but tighter standard deviation. This is because the surrogate is a stronger constraint, i.e., the surrogate is a strict bound of the original [cf. \eqref{eq:resultConstrainedVariance}].

Lastly, we evaluate the effects of hyper-parameters on the variance-constrained learning, i.e., the GRES($p,q$) realizations $N$ for empirical estimations [cf. \eqref{eq:averageApproximate1}] in Fig. \ref{fig_different_samples}, the gradient steps $\Gamma$ at the primal phase [cf. \eqref{eq_priup}] in Fig. \ref{fig_different_steps} and the number of dropping edges $M_d$ in Fig. \ref{fig_different_edges}. Fig. \ref{fig_different_samples} shows that the expected cost fluctuates with $N$ and steadies as $N$ becomes large, while the standard deviation decreases with $N$. This is because the empirical estimation with a larger $N$ approximates better the expectation, which however takes more training time. 
Fig. \ref{fig_different_steps} shows that the variance-constrained learning converges faster and to a lower value as $\Gamma$ increases. This corroborates Theorem \ref{thm:convergencePrimalDual} since more gradient steps approach better the optimal solution at each primal phase, which reduces the error size $\xi$ and accelerates the convergence. It is remarkable from Figs. \ref{fig_different_samples}-\ref{fig_different_steps} that small values of graph realizations, e.g., $N \ge 10$, and gradient steps, e.g., $\Gamma \ge 1$, achieve a satisfactory performance, indicating an efficient implementation. Fig. \ref{fig_different_edges} shows that the expected accuracy decreases and the standard deviation increases with $M_d$. This follows our finding in Theorem \ref{thm:varianceAnalysis} that 
more unstable edges increase the graph stochasticity and the latter degrades the performance. Finally, we corroborate the relation between the variance and the discrimination power analyzed in Sec. \ref{DOL}. Fig. \ref{fig_Lipschitz} shows that the expected Lipschitz constant $C_L$ of stochastic graph filters increases with the variance bound $C_v$. This corresponds to the theoretical finding in Theorem \ref{thm:varianceAnalysis} that constraining the variance may lead to a less discriminative architecture, which contains filters with flatter frequency responses.

\subsection{Recommender Systems}

We consider the recommender system (RecSys) with data from MovieLens100k, which comprises $943$ users and $1682$ movies \cite{harper2015movielens}. Following the pre-processing steps in \cite{monti2017geometric}, we build the graph by considering nodes as movies and edges as similarities between them. We compute the movie similarity via the Pearson correlation and keep the $35$ edges with the highest correlation. The graph signal is the ratings given by a user to the movies, where the signal value is zero if the movie is unrated.

In the RecSys, accuracy measures how well we predict the ratings a user has given to the movies. However, high accuracy is not necessarily linked to a better user satisfaction. Diversity also plays an important role, which measures the capability of the RecSys to include items of different categories in the recommendation list \cite{aggarwal2016recommender}. To measure accuracy we use the root mean squared error (RMSE), which is a standard criterion for the rating-based RecSys. To measure diversity we use the aggregated diversity for the recommendation list containing top ten items (AD@10), which is defined as the number of different items included in the list. A lower RMSE indicates a better accuracy and a higher AD implies a more diversified RecSys, i.e., the system does not overfit accuracy by recommending only niche items. The joint goal is to tweak the accuracy-diversity trade-off, i.e., predict accurate ratings and increase the recommendation diversity.

\smallskip
\noindent \textbf{Parameterization.} We consider the SGNN comprising a single layer with $F=32$ filters of order $K=4$ and the Leaky ReLU nonlinearity. The graph stochasticity throughout the architecture is leveraged as a training strategy to aid diversity because it will randomly remove some similarity edges between movies and connect different movies with each other \cite{monti2017geometric}. We consider the first $35$ edges with the highest correlation may be dropped and the next $20$ edges may be added with a probability $p$, corresponding to the GRES($p,q$) model with $p=q$ for simplicity. The constraint bounds are set as $C_f = 0$ and $C_s = 0.5$.

\smallskip
\noindent \textbf{Performance.} We compare the accuracy-diversity trade-off of the vanilla GNN, the SGNN with and without the variance-constrained learning. Fig. \ref{subfig:recommendationRMSE} shows the expected RMSE and the standard deviation under different edge dropping / adding probabilities $p \in [0.05, 0.25]$. For a lower $p \to 0$, the graph is stable and the SGNN exhibits comparable accuracies to the GNN; for a higher $p$, the graph varies more dramatically and the SGNN degrades gradually. The variance-constrained learning accounts for the variance during training, and thus maintains a lower standard deviation around the expected RMSE. Contrarily, the baseline method ignores this factor and has a higher standard deviation that increases with $p$.

Figs. \ref{subfig:recommendationAD}-\ref{subfig:recommendationAD2} display the expected AD@10 and the standard deviation around it. The SGNN improves the diversity compared to the GNN, which can be explained by the involved graph stochasticity. While restricting the variance during training, the variance-constrained learning achieves a comparable (slightly lower) AD@10 to the baseline method. This result together with the well-controlled RMSE in Fig. \ref{subfig:recommendationRMSE} indicate that the variance-constrained learning exhibits a better accuracy-diversity trade-off.


\section{Conclusions} \label{sec_conclusions}

We proposed a variance-constrained learning strategy for stochastic graph neural networks that achieves a trade-off between the expected performance and stochastic deviations. This strategy adheres to solving a constrained stochastic optimization problem. We developed a primal-dual learning method to solve the problem in the dual domain, which alternates gradient updates between the SGNN parameters and the dual variable. The variance-constrained learning can be interpreted as a self-learning variance regularizer that provides explicit guarantees for stochastic deviations. A statistical analysis on the SGNN output is conducted to identify how 
the output variance is decreased and indicates the constrained variance comes at the expense of the discrimination power. We further analyzed the duality gap of the variance-constrained optimization problem and the convergence of the primal-dual learning method, which characterize the solution suboptimality and provide theoretical guarantees for the performance. Numerical results corroborate that the variance-constrained learning finds a favorable balance between the optimal performance and the deviation degradation.








\appendices 



\section{Proof of Proposition \ref{thm:probabilityBound}} \label{pr:proofTheorem1}

\noindent To ease notation, denote by $\bbPhi=\bbPhi(\bbx;\bbS_{P:1},\ccalH)$ and by $\bar{\bbPhi} = \mathbb{E}_\ccalM[\bbPhi(\bbx;\bbS_{P:1},\ccalH)]$. For a feasible solution $\ccalH$, it holds that
\begin{align}\label{eq:proofThm1eq1}
	{\rm Var} \left[\bbPhi\right] = \frac{1}{n}\mathbb{E}_\ccalM \big[\big\| \bbPhi - \bar{\bbPhi}\big\|^2\big] \le C_v.
\end{align} 
Leveraging the conditional probability, we represent ${\rm Var} \left[\bbPhi\right]$ as
\begin{align}\label{eq:proofThm1eq2}
	&\mathbb{E}_\ccalM\! \left[\frac{1}{n}\big\| \bbPhi \!-\! \bar{\bbPhi}\big\|^2 \Big| \frac{1}{n}\big\| \bbPhi \!-\! \bar{\bbPhi}\big\|^2 \!\!\le\! \epsilon \right] \!\!\cdot\! \text{Pr} \left[\! \frac{1}{n}\big\| \bbPhi \!-\! \bar{\bbPhi}\big\|^2 \!\!\le\! \epsilon \!\right]\\
	&\!+\!\mathbb{E}_\ccalM\! \left[\frac{1}{n}\big\| \bbPhi \!-\! \bar{\bbPhi}\big\|^2\Big| \frac{1}{n}\big\| \bbPhi \!-\! \bar{\bbPhi}\big\|^2 \!\!>\! \epsilon \right] \!\!\cdot\! \text{Pr} \left[ \!\frac{1}{n}\big\| \bbPhi \!-\! \bar{\bbPhi}\big\|^2 \!\!>\! \epsilon \right]\!.\nonumber
\end{align}
Since $\big\| \bbPhi - \bar{\bbPhi}\big\|^2 \ge 0$, we can lower bound \eqref{eq:proofThm1eq2} by
\begin{align}\label{eq:proofThm1eq3}
	&\quad 0 \cdot \text{Pr} \left[ \frac{1}{n}\big\| \bbPhi - \bar{\bbPhi}\big\|^2 \!\le\! \epsilon \right] \!+\! \epsilon \cdot \text{Pr} \left[ \frac{1}{n}\big\| \bbPhi - \bar{\bbPhi}\big\|^2 \!>\! \epsilon \right]. 
\end{align}
By substituting \eqref{eq:proofThm1eq3} into \eqref{eq:proofThm1eq1}, we get $\epsilon \cdot \text{Pr}\! \left[ \big\| \bbPhi \!-\! \bar{\bbPhi}\big\|^2/n \!>\! \epsilon \right] \!<\! C_v$. Since $\text{Pr} \big[ \big\| \bbPhi - \bar{\bbPhi}\big\|^2 / n \!>\! \epsilon \big] \!+\! \text{Pr} \big[ \big\| \bbPhi \!-\! \bar{\bbPhi}\big\|^2 / n \!\le\! \epsilon \big] \!=\! 1$, we have
\begin{align}\label{eq:proofCoro2_5}
	\text{Pr}\left[\frac{1}{n}\big\| \bbPhi - \bar{\bbPhi}\big\|^2 \le \epsilon \right] \ge 1- \frac{C_v}{\epsilon}.
\end{align}


\section{Proof of Theorem \ref{thm:varianceAnalysis}} \label{pr:proofTheorem2}

\noindent We start by considering the variance of the filter output $\bbu = \bbH(\bbS_{K:0})\bbx$. Denote by $\bbS_k = \barbS + \bbE_k$ with $\barbS$ the expected shift operator and $\bbE_k$ the random deviation. Following the proof of Proposition 1 in \cite{gao2021stochastic}, substituting the filter output \eqref{eq:stochasticGraphFilterOutput} into the variance \eqref{eq:varianceDefinition} and expanding the terms yields
\begin{align} \label{eq:proofThm2eq1}
	&{\rm Var}[\bbu] =\frac{1}{n}\sum_{k=0}^K \sum_{\ell=0}^K h_k h_\ell \tr \left( \mathbb{E}\left[ \bbC_{k\ell}\right] \right)\\
	& \! +\!\frac{1}{n}\sum_{k=1}^K \sum_{\ell=1}^K h_k h_\ell \tr \Big( \mathbb{E}\Big[\! \sum_{r=1}^{\lfloor k\ell \rfloor} \barbS^{k-r}\bbE_r \barbS^{r\!-\!1}\bbx \bbx^\top \barbS^{r-1}\bbE_r \barbS^{\ell-r}\!\Big] \Big)\nonumber
\end{align}
where $\lceil k\ell \rceil = \max(k,\ell)$, $\lfloor k\ell \rfloor = \min(k,\ell)$ and $\bbC_{k\ell}$ is the sum of the terms that contain at least two deviations $\bbE_{r_1}$ and $\bbE_{r_2}$ with $r_1 \ne r_2$. We now proceed by analyzing each of the terms in \eqref{eq:proofThm2eq1} starting from the latter.

\smallskip
\noindent $\textbf{Second term.}$ The second term in \eqref{eq:proofThm2eq1} is similar to the second term in \cite[Eq. (37)]{gao2021stochastic}. Following steps Eq. (39)-(43) in \cite{gao2021stochastic}, we can upper bound it by
\begin{align} \label{eq:proofThm2eq4}
& \sum_{i = 1}^N\hat{x}_i^2 \sum_{r\!=\!1}^K\!\tr \Big( \sum_{k, \ell=r}^K h_k h_\ell \bar{\lambda}_i^{2r-2} \barbS^{k+\ell-2r} \mathbb{E}\!\left[ \bbE_r^2\right] \Big)
\end{align}
where $\{\hat{x}_i\}_{i=1}^n$ are the Fourier coefficients of $\bbx$ over $\barbS$. From Lemma \ref{lemma:errorMatrixExpectation} in Appendix \ref{lemmasProof}, we have $\mathbb{E}\left[ \bbE_k^2 \right] =  \alpha p(1-p) \bbE_d + \alpha q(1-q) \bbE_a$ with (i) $\alpha = 1$, $\bbE_d = \bbD_d$, $\bbE_a = \bbD_a$ the degree matrices of $\ccalG_d$, $\ccalG_a$ if $\bbS = \bbA$ is the adjacency matrix and (ii) $\alpha = 2$, $\bbE_d = \bbL_d$, $\bbE_a = \bbL_a$ the Laplacian matrices of $\ccalG_d$, $\ccalG_a$ if $\bbS = \bbL$ is the Laplacian matrix. Then, using the trace property $\tr(\bbA \bbB) \le \| \bbA \| \tr(\bbB)$ for any square matrix $\bbA$ and positive semi-definite matrix $\bbB$ and substituting $\mathbb{E}\left[ \bbE_k^2 \right]$, we bound \eqref{eq:proofThm2eq4} by
\begin{align}
\label{eq:proofThm2eq5}
& 2 p(1-p)\!\sum_{i = 1}^N\!\hat{x}_i^2 \big\| \sum_{r\!=\!1}^K\! \sum_{k, \ell=r}^K\! h_k h_\ell \bar{\lambda}_i^{2r-2} \barbS^{k+\ell-2r} \big\| \tr\! \left( \bbE_d \right)\\
& + 2 q(1-q)\sum_{i = 1}^N\hat{x}_i^2 \big\| \sum_{r\!=\!1}^K\! \sum_{k, \ell=r}^K\! h_k h_\ell \bar{\lambda}_i^{2r-2} \barbS^{k+\ell-2r} \big\| \tr\! \left( \bbE_a \right)\nonumber
\end{align}
with $\tr\left( \bbE_d \right)\!=\! 2M_d$ and $\tr\left( \bbE_a \right)\!=\! 2M_a$ where $M_d$ and $M_a$ are the numbers of dropping and adding edges [Def. \ref{def_res}]. The matrix norm in \eqref{eq:proofThm2eq5} is similar to that in \cite[Eq. (44)]{gao2021stochastic}. In this context, following steps Eq. (45)-(49) of \cite{gao2021stochastic} and using the Lipschitz property of $h(\bblambda)$ [As. \ref{as:LipschitzFilter}], we can upper bound it by $K C_L^2$. By substituting this bound into \eqref{eq:proofThm2eq5}, the second term in \eqref{eq:proofThm2eq1} is bounded as
\begin{align}
\label{eq:proofThm2eq9} &\frac{1}{n}\mathbb{E}\Big[ \sum_{k, \ell=1}^K h_k h_\ell \sum_{r=1}^{\lfloor k\ell \rfloor} \tr \left( \barbS^{k-r}\bbE_r \barbS^{r-1}\bbx \bbx^\top \barbS^{r-1}\bbE_r \barbS^{\ell-r} \right) \Big]\nonumber \\
& \!\le\! \frac{ 4 K}{n} C_L^2 \big(M_d p (1-p) + M_a q (1-q) \big) \| \bbx \|^2.
\end{align}

\noindent $\textbf{First term.}$ Matrix $\bbC_{k\ell}$ is the sum of the remaining expansion terms. Each of these terms contains at least two deviations $\bbE_k$, $\bbE_\ell$ with $k \neq \ell$ and can be bounded by a factor containing at least two terms $\tr\left(\mathbb{E}[\bbE_{k}^2]\right)$ and $\tr\left(\mathbb{E}[\bbE_{\ell}^2]\right)$. Since the filter coefficients $\{ h_k \}_{k=0}^K$ and the expected shift operator norm $\|\barbS\|$ are bounded, we can write the first term in \eqref{eq:proofThm2eq1} as
\begin{gather} \label{eq:proofThm2eq10}
\begin{split}
 \mathbb{E} \Big[\!\sum_{k, \ell\!=\!0}^K h_k h_\ell \bbC_{k\ell} \Big] \!=\! \ccalO(p^2(1\!-\!p)^2) + \ccalO(q^2(1\!-\!q)^2).
\end{split}
\end{gather}

\noindent By substituting \eqref{eq:proofThm2eq9} and \eqref{eq:proofThm2eq10} into \eqref{eq:proofThm2eq1}, we have
\begin{gather} \label{eq:proofThm2eq11}
\begin{split}
 {\rm Var} \left[ \bbu \right] &\le \frac{4 K}{n} \big(M_d p (1-p) + M_a q (1-q) \big) C_L^2 \| \bbx \|^2\\
&+ \ccalO(p^2(1\!-\!p)^2) + \ccalO(q^2(1\!-\!q)^2).
\end{split}
\end{gather}

Then, by leveraging \eqref{eq:proofThm2eq11} and steps Eq. (54)-(73) in the proof of Theorem 1 in \cite{gao2021stochastic}, we extend the variance bound from the filter to the SGNN and obtain the result \eqref{eq:varianceAnalysis} to complete the proof. 


\section{Proof of Theorem \ref{thm:dualityGap}} \label{pr:proofTheorem3}

We prove the theorem as follows. First, we upper bound the primal solution $\mathbb{P}$ of problem \eqref{eq:alternativeVarianceConstrainedProblem} using the primal solution $\widetilde{\mathbb{P}}$ of the general problem \eqref{eq:generalizedVarianceConstrainedProblem} [cf. \eqref{eq:proofThm3eq14}]. Then, we lower bound the dual solution $\mathbb{D}$ of problem \eqref{eq:alternativeVarianceConstrainedProblem} using the dual solution $\widetilde{\mathbb{D}}$ of the general problem \eqref{eq:generalizedVarianceConstrainedProblem} [cf. \eqref{eq:proofThm3eq19}]. Lastly, we complete the proof by leveraging the strong duality $\widetilde{\mathbb{P}} = \widetilde{\mathbb{D}}$ of the general problem \eqref{eq:generalizedVarianceConstrainedProblem}. 

\smallskip
\noindent \textbf{Primal upper bound.} For the primal upper bound, we first particularize the function $\widetilde{f}(\bbx, \widetilde{\bbS}_{P:1})$ to the SGNN $\bbPhi(\bbx;\widetilde{\bbS}_{P:1}, \ccalH)$ and then particularize the continuous distribution $\widetilde{\ccalM}$ to the discrete one $\ccalM$. We do so for the objective and constraints, respectively. 

\textit{Objective:} Denote by $\ccalC_\ccalT(\cdot, \cdot) = \mathbb{E}_\ccalT [\ccalC(\cdot, \cdot)]$, $\widetilde{f} = \widetilde{f}(\bbx;\widetilde{\bbS}_{P:1})$, $f = f(\bbx;\bbS_{P:1})$, $\widetilde{\bbPhi}=\bbPhi(\bbx;\widetilde{\bbS}_{P:1}, \!\ccalH)$ and $\bbPhi = \bbPhi(\bbx;\bbS_{P:1}, \!\ccalH)$ for concise notations. By using Jensen's inequality combined with the fact that the absolute value $|\cdot|$ is a convex function, we can bound the deviation induced by the SGNN parameterization as
\begin{align}\label{eq:proofThm3eq2}
	&\!\!\Big| \mathbb{E}_{\widetilde{\ccalM}}\! \Big[ \ccalC_\ccalT(\bby,\! \widetilde{f})\Big] \!\!-\! \mathbb{E}_{\widetilde{\ccalM}}\!\Big[ \ccalC_\ccalT(\bby,\! \widetilde{\bbPhi}\!) \Big] \!\Big| \!\!\le\! \mathbb{E}_{\widetilde{\ccalM}}\! \Big[ \big| \ccalC_\ccalT(\bby,\! \widetilde{f}) \!-\!  \ccalC_\ccalT(\bby,\! \widetilde{\bbPhi}\!) \big| \Big]\!.
\end{align}
By using the Lipschitz condition of the loss function [As. \ref{as:LipschitzCost}], Jensen's inequality with the fact that the square $(\cdot)^2$ is a convex function, and the $\eps$-universal parameterization [As. \ref{as:universalParameterization}], we can further bound \eqref{eq:proofThm3eq2} as
 \begin{align}\label{eq:proofThm3eq25}
 	&\mathbb{E}_{\widetilde{\ccalM}} \big[ \big| \ccalC_\ccalT(\bby, \widetilde{f}) -  \ccalC_\ccalT(\bby, \widetilde{\bbPhi}) \big| \big]\le C_\ell \mathbb{E}_{\widetilde{\ccalM}} \big[ \| \widetilde{f} - \widetilde{\bbPhi} \| \big] \\
 	& \le C_\ell \sqrt{\mathbb{E}_{\widetilde{\ccalM}} \big[ \| \widetilde{f} - \widetilde{\bbPhi} \|^2 \big]} \le C_L \epsilon.\nonumber
 \end{align}
Since $\widetilde{\ccalM}$ is the $\varepsilon$-Borel generalization of $\ccalM$, we use the stability result of Lemma \ref{lemma:stabilitySGNN} in Appendix \ref{lemmasProof} and bound the deviation induced by the distribution generalization as
\begin{align}\label{eq:proofThm3eq4}
	&\!\big| \mathbb{E}_{\widetilde{\ccalM}} \big[ \ccalC_\ccalT(\bby, \widetilde{\bbPhi})\big] \!\!-\! \mathbb{E}_{\ccalM}\!\left[ \ccalC_\ccalT(\bby, \bbPhi) \right] \!\big| \le C_\ell C_B \varepsilon + \ccalO(\varepsilon^2).
\end{align}
with $C_B$ the stability constant. Adding and subtracting $\mathbb{E}_{\widetilde{\ccalM}}\big[ \ccalC_\ccalT(\bby, \widetilde{\bbPhi}) \big]$ in $\mathbb{E}_{\widetilde{\ccalM}} \big[ \ccalC_\ccalT(\bby, \widetilde{f})\big] - \mathbb{E}_{\ccalM}\!\left[ \ccalC_\ccalT(\bby, \bbPhi) \right]$, and using the triangular inequality, \eqref{eq:proofThm3eq25} and \eqref{eq:proofThm3eq4}, we get
\begin{align}\label{eq:proofThm3eq5}
	&\big| \mathbb{E}_{\widetilde{\ccalM}} \big[ \ccalC_\ccalT(\bby, \widetilde{f})\big] \!-\! \mathbb{E}_{\ccalM}\!\left[ \ccalC_\ccalT(\bby, \bbPhi) \right] \big|  \!\le\! C_\ell \epsilon \!+\! C_\ell C_B \varepsilon \!+\! \ccalO(\varepsilon^2).
\end{align}

\textit{First order moment constraint:} Following similar steps, we can bound the constraint deviation induced by the SGNN parameterization as
\begin{align}\label{eq:proofThm3eq6}
	&\Big| \frac{1}{n}\mbE_{\widetilde{\ccalM}} \Big[\! \sum_{i=1}^n [\widetilde{f}]_i \Big]  \!-\! \frac{1}{n}\mbE_{\widetilde{\ccalM}} \Big[\! \sum_{i=1}^n [\widetilde{\bbPhi}]_i \Big] \Big| \\
	&\le \frac{1}{n}\mbE_{\widetilde{\ccalM}} \Big[\! \sum_{i=1}^n \big| [\widetilde{f}]_i \!-\!  [\widetilde{\bbPhi}]_i \big| \Big] = \frac{1}{n}\mathbb{E}_{\widetilde{\ccalM}} \Big[ \| \widetilde{f} \!-\! \widetilde{\bbPhi} \|_1 \Big] \le  \frac{\eps}{\sqrt{n}}\nonumber
\end{align}
where $\|\cdot\|_1$ is 1-norm and where the norm property $\|\cdot\|_1 \le \sqrt{n}\|\cdot\|$ and the $\eps$-universal parameterization are used in the last inequality. Likewise, we can bound the constraint deviation induced by the distribution generalization as
\begin{align}\label{eq:proofThm3eq7}
	&\!\Big| \frac{1}{n}\mbE_{\widetilde{\ccalM}}\! \Big[ \sum_{i=1}^n [\widetilde{\bbPhi}]_i \Big] \!-\! \frac{1}{n}\mbE_{\ccalM}\! \Big[ \sum_{i=1}^n [\bbPhi]_i \Big] \Big| \!\le\!  \frac{C_B}{\sqrt{n}}\varepsilon \!+\! \ccalO(\varepsilon^2)
\end{align}
which holds similarly because of $\|\cdot\|_1 \le \sqrt{n}\|\cdot\|$ and the stability result in Lemma \ref{lemma:stabilitySGNN}. By using \eqref{eq:proofThm3eq6} and \eqref{eq:proofThm3eq7}, we get
\begin{align}\label{eq:proofThm3eq8}
	&\!\Big| \frac{1}{n}\mbE_{\widetilde{\ccalM}} \Big[ \sum_{i=1}^n [\widetilde{f}]_i \Big] \!-\! \frac{1}{n}\mbE_{\ccalM} \Big[ \sum_{i=1}^n [\bbPhi]_i \Big] \Big|\!\le\! \frac{1}{\sqrt{n}}\eps \!+\! \frac{C_B}{\sqrt{n}}\varepsilon \!+\! \ccalO(\varepsilon^2).
\end{align}

\textit{Second order moment constraint:} Since $a^2 - b^2 = (a+b)(a-b) = (2b + a - b)(a-b)$ for any $a,b \in \mathbb{R}$, we have the bound for the SGNN parameterization as
\begin{align}\label{eq:proofThm3eq9}
	&\Big| \frac{1}{n} \mbE_{\widetilde{\ccalM}} \Big[ \sum_{i=1}^n [\widetilde{f}]^2_i \Big] \!-\! \frac{1}{n} \mbE_{\widetilde{\ccalM}}\! \Big[ \sum_{i=1}^n [\widetilde{\bbPhi}]^2_i \Big] \Big| \\
	&=\! \frac{1}{n}\Big| \mbE_{\widetilde{\ccalM}}\! \Big[\! \sum_{i=1}^n (2[\widetilde{\bbPhi}]_i + [\widetilde{f}]_i \!-\! [\widetilde{\bbPhi}]_i)([\widetilde{f}]_i \!-\! [\widetilde{\bbPhi}]_i) \Big] \Big| \nonumber \\
	& \le \! \frac{1}{n} \mbE_{\widetilde{\ccalM}} \Big[\big| \sum_{i=1}^n\! 2[\widetilde{\bbPhi}]_i([\widetilde{f}]_i \!-\! [\widetilde{\bbPhi}]_i) \big| \!+\! \sum_{i=1}^n\!([\widetilde{f}]_i \!-\! [\widetilde{\bbPhi}]_i)^2 \Big] \nonumber
\end{align}
where the last inequality holds because of the triangular inequality. Using $|[\widetilde{\bbPhi}]_i| \le \|\widetilde{\bbPhi}\| \le C_y$ [As. \ref{as:inputBound}], $\|\cdot\|_1 \le \sqrt{n}\|\cdot\|$, and the $\eps$-universal parameterization, we can upper bound \eqref{eq:proofThm3eq9} by
\begin{align}\label{eq:proofThm3eq95}
&\frac{2C_y}{n} \mbE_{\widetilde{\ccalM}} \big[ \|\widetilde{f} \!-\! \widetilde{\bbPhi}\|_1 \big] \!+\!\frac{1}{n} \mbE_{\widetilde{\ccalM}} \big[ \|\widetilde{f} \!-\! \widetilde{\bbPhi}\|_2^2\big] \!\le\!  \frac{2C_y}{\sqrt{n}}\eps \!+\! \frac{1}{n}\eps^2.
\end{align}
By similarly leveraging $a^2 - b^2 = (a+b)(a-b)$ for any $a,b \in \mathbb{R}$, $|[\widetilde{\bbPhi}]_i| \le \|\widetilde{\bbPhi}\| \le C_y$, $|[\bbPhi]_i| \le \|\bbPhi\| \le C_y$, and the stability results in Lemma \ref{lemma:stabilitySGNN}, we have the bound for the distribution generalization as
\begin{align}\label{eq:proofThm3eq10}
	&\Big| \frac{1}{n}\mbE_{\widetilde{\ccalM}} \Big[ \sum_{i=1}^n [\widetilde{\bbPhi}]^2_i \Big] \!-\!\frac{1}{n}\mbE_{\ccalM}\! \Big[ \sum_{i=1}^n [\bbPhi]^2_i \Big] \Big|\le \frac{2C_yC_B}{\sqrt{n}}\varepsilon + \ccalO(\varepsilon^2). 
\end{align}
By using \eqref{eq:proofThm3eq95} and \eqref{eq:proofThm3eq10}, we get
\begin{align}\label{eq:proofThm3eq11}
	&\Big| \frac{1}{n}\mbE_{\widetilde{\ccalM}}\! \Big[ \sum_{i=1}^n [\widetilde{f}]^2_i \Big] - \frac{1}{n}\mbE_\ccalM\! \Big[ \sum_{i=1}^n [\bbPhi]^2_i \Big] \Big|\nonumber \\
	&\le \big(\frac{2C_y}{\sqrt{n}} + \frac{\eps}{n}\big)\eps + \frac{2C_yC_B}{\sqrt{n}}\varepsilon + \ccalO(\varepsilon^2). 
\end{align}

We now consider a modified version of problem \eqref{eq:generalizedVarianceConstrainedProblem} that changes the constraints with the bounds in \eqref{eq:proofThm3eq8} and \eqref{eq:proofThm3eq11}
\begin{alignat}{3} \label{eq:proofThm3eq12}
	\widetilde{\mathbb{P}}_{\epsilon, \varepsilon}:=   &\min_{\widetilde{f}} \mathbb{E}_{\widetilde{\ccalM}} \Big[ \ccalC_\ccalT\big(\bby, \widetilde{f}\big) \Big],             \\
	&  \st ~ \frac{1}{n}\mbE_{\widetilde{\ccalM}} \! \Big[ \sum_{i=1}^n [\widetilde{f}]_i \Big] \!\ge\! C_f \!+\! \frac{\eps}{\sqrt{n}} \!+\! \frac{C_B}{\sqrt{n}}\varepsilon\! +\! \ccalO(\varepsilon^2)  \nonumber\\
	& \!\frac{1}{n}\mbE_{\widetilde{\ccalM}}\! \Big[ \!\sum_{i=1}^n [\widetilde{f}]^2_i \Big] \!\!\le\! C_s\!-\! \big(\frac{2C_y}{\sqrt{n}} \!+\! \frac{\eps}{n}\big)\eps \!-\! \frac{2C_yC_B}{\sqrt{n}}\varepsilon \!-\! \ccalO(\varepsilon^2\!).  \nonumber
\end{alignat}
For each feasible function $\widetilde{f}(\bbx;\widetilde{\bbS}_{P:1})$ in \eqref{eq:proofThm3eq12}, there exist a set of $\ccalH$ and $\bbS_{P:1}\in \ccalM$ such that the corresponding SGNN parameterization $\bbPhi(\bbx;\bbS_{P:1}, \ccalH)$ satisfies the constraints in \eqref{eq:alternativeVarianceConstrainedProblem} because of the constraint gaps established in \eqref{eq:proofThm3eq8} and \eqref{eq:proofThm3eq11}. This observation and the objective gap in \eqref{eq:proofThm3eq4} imply that the primal solution $\mathbb{P}$ of problem \ref{eq:alternativeVarianceConstrainedProblem} is close to the primal solution $\widetilde{\mathbb{P}}_{\epsilon, \varepsilon}$ of this modified problem \ref{eq:proofThm3eq12} by at most
\begin{align}\label{eq:proofThm3eq13}
	\mathbb{D} \le \mathbb{P} \le \widetilde{\mathbb{P}}_{\epsilon, \varepsilon} + C_\ell \epsilon + C_\ell C_B \varepsilon + \ccalO(\varepsilon^2).
\end{align}
By further using the perturbation inequality between $\widetilde{\mathbb{P}}_{\epsilon, \varepsilon}$ and $\widetilde{\mathbb{P}}$ \cite[Eq. (5.57)]{boyd2004convex}, we get
\begin{align}\label{eq:proofThm3eq14}
	\mathbb{D} \le \mathbb{P} &\le \widetilde{\mathbb{P}} \!+\! \Big(C_\ell \! + \!  \frac{\widetilde{\gamma}_1^*}{\sqrt{n}} \!+\! \widetilde{\gamma}_2^*\big(\frac{2C_y}{\sqrt{n}} \!+\! \frac{\eps}{n}\big) \Big) \eps \nonumber\\	
&\quad +	\Big( C_\ell C_B \!+\! \frac{\widetilde{\gamma}_1^* C_B }{\sqrt{n}} \!+\! \frac{2 \widetilde{\gamma}_2^* C_y C_B}{\sqrt{n}} \Big) \varepsilon + \ccalO(\varepsilon^2)
\end{align}
with $\widetilde{\bbgamma}^* = [\widetilde{\gamma}^*_1, \widetilde{\gamma}^*_2]$ the optimal dual variable of problem \eqref{eq:generalizedVarianceConstrainedProblem}. 

\smallskip
\noindent \textbf{Dual lower bound.} Following \eqref{eq:Lagrangian}, we can represent the Lagrangian of problem \eqref{eq:generalizedVarianceConstrainedProblem} at its optimal dual variable $\widetilde{\bbgamma}^*$ as
\begin{align} \label{eq:proofThm3eq17}
	\widetilde{\ccalL}(\widetilde{f}, \widetilde{\bbgamma}^*) = \mathbb{E}_{\widetilde{\ccalM}}\! \Big[ \ccalC_\ccalT(\bby, \widetilde{f})\Big]\!&+ \widetilde{\gamma}_1^* \Big(\! C_f-\!\frac{1}{n}\mbE_{\widetilde{\ccalM}} \Big[\! \sum_{i=1}^n [\widetilde{f}]_i \Big] \Big) \nonumber\\
	& \!\!\!\!\!\!\!- \widetilde{\gamma}_2^* \Big( C_s-\frac{1}{n}\mbE_{\widetilde{\ccalM}} \Big[ \sum_{i=1}^n [\widetilde{f}]^2_i \Big] \Big).
\end{align}
We first particularize the continuous distribution $\widetilde{\ccalM}$ to the discrete distribution $\ccalM$ in \eqref{eq:proofThm3eq17} and obtain the particularized Lagrangian $\ccalL(f, \widetilde{\bbgamma}^*)$. For any function $f$ and $\bbS_{P:1} \in \ccalM$, we can define a function $\widetilde{f}$ as $\widetilde{f}(\bbx;\widetilde{\bbS}_{P:1}) = f(\bbx;\bbS_{P:1})$ for any $\widetilde{\bbS}_{P:1} \in \ccalB_{\varepsilon}(\bbS_{P:1}) \subset \widetilde{\ccalM}$ [Def. \ref{def1:BorelSet}]. Since both $\widetilde{\ccalM}$ and $\ccalM$ are uniform distributions, substituting $f$ and $\widetilde{f}$ into their respective Lagrangian yields $\ccalL(f, \widetilde{\bbgamma}^*) = \widetilde{\ccalL}(\widetilde{f}, \widetilde{\bbgamma}^*)$. That is, for any $f$, there exists an associated $\widetilde{f}$ satisfying $\ccalL(f, \widetilde{\bbgamma}^*) = \widetilde{\ccalL}(\widetilde{f}, \widetilde{\bbgamma}^*)$. Therefore, we have
\begin{align} \label{eq:proofThm3eq18}
	\min_{f} \ccalL(f, \widetilde{\bbgamma}^*) \ge \min_{\widetilde{f}} \widetilde{\ccalL}(\widetilde{f}, \widetilde{\bbgamma}^*). 
\end{align}
We then particularize the function $f$ to the SGNN $\bbPhi$ and obtain the Lagrangian $\ccalL(\ccalH, \widetilde{\bbgamma}^*)$ [cf. \eqref{eq:Lagrangian}]. Since the set of functions spanned by the SGNN $\bbPhi$ is a subset of the set of functions spanned by $f$, it holds that 
\begin{align} \label{eq:proofThm3eq19}
	\min_{\ccalH} \ccalL(\ccalH, \widetilde{\bbgamma}^*) \ge \min_{f} \ccalL(f, \widetilde{\bbgamma}^*) \ge \min_{\widetilde{f}} \widetilde{\ccalL}(\widetilde{f}, \widetilde{\bbgamma}^*). 
\end{align}
By substituting the facts $\mathbb{D} = \max_{\bbgamma} \min_{\ccalH} \ccalL(\ccalH, \bbgamma) \ge \min_{\ccalH} \ccalL(\ccalH, \widetilde{\bbgamma}^*)$ and $\widetilde{\mathbb{D}} = \min_{\widetilde{f}} \widetilde{\ccalL}(\widetilde{f}, \widetilde{\bbgamma}^*)$ into \eqref{eq:proofThm3eq19} and leveraging the strong duality $\tilde{\mathbb{P}} = \tilde{\mathbb{D}}$ in Proposition \ref{prop:strongDualityGeneralized}, we get $\mathbb{D} \ge \widetilde{\mathbb{D}} = \widetilde{\mathbb{P}}$.

By combining the primal upper bound in \eqref{eq:proofThm3eq14} and the dual lower bound $\mathbb{D} \ge \widetilde{\mathbb{P}}$, we complete the proof
\begin{align}\label{proof:prop3eq20}
	|\mathbb{P}-\mathbb{D}| \le & \Big( C_\ell + \frac{\widetilde{\gamma}_1^*}{\sqrt{n}} \!+\! \widetilde{\gamma}_2^*\big(\frac{2C_y}{\sqrt{n}} \!+\! \frac{\eps}{n}\big) \Big) \eps\\	
	& +	\Big( C_\ell C_B \!+\! \frac{\widetilde{\gamma}_1^* C_B }{\sqrt{n}} \!+\! \frac{2 \widetilde{\gamma}_2^* C_y C_B}{\sqrt{n}} \Big) \varepsilon + \ccalO(\varepsilon^2). \nonumber
\end{align}


\section{Proof of Theorem \ref{thm:convergencePrimalDual}} \label{pr:theorem4}

\noindent Let $\bbgamma^*$ be the optimal dual variable for the dual problem [cf. \eqref{eq:dualFunction}]. From the dual update \eqref{eq:dualUpdate} and the fact that the non-negative projection $[\cdot]_+$ is non-expansive, we can write
\begin{align}\label{eq:proofThm4eq225}
	&\| \bbgamma_{t+1} - \bbgamma^* \|^2 \\
	&\le \Big( \gamma_{1,t} + \eta_\gamma\Big(\! C_f\!-\!\frac{1}{n}{\mathbb{E}}_\ccalM\! \Big[ \!\sum_{i=1}^n [\bbPhi(\bbx;\bbS_{P:1},\ccalH_{t+1})]_i \!\Big] \Big) \!-\! \gamma^*_1 \Big)^2 \nonumber \\
	& + \Big( \gamma_{2,t} \!-\! \eta_\gamma\Big(\! C_s\!-\!\frac{1}{n}{\mathbb{E}}_\ccalM \Big[ \!\sum_{i=1}^n [\bbPhi(\bbx;\bbS_{P:1},\ccalH_{t+1})]^2_i \Big] \Big) \!-\! \gamma^*_2 \Big)^2. \nonumber
\end{align}
By expanding the square operations $(\cdot)^2$ in \eqref{eq:proofThm4eq225}, we can rewrite the upper bound as 
\begin{align}\label{eq:proofThm4eq25}
	&\| \bbgamma_{t} \!-\! \bbgamma^* \|^2 \!+\! \eta_\gamma^2 \Big( C_f\!-\!\frac{1}{n}{\mathbb{E}}_\ccalM\! \Big[ \!\sum_{i=1}^n [\bbPhi(\bbx;\bbS_{P:1},\ccalH_{t+1})]_i \Big] \Big)^2 \nonumber\\
	&+ \eta_\gamma^2 \Big( C_s\!-\!\frac{1}{n}{\mathbb{E}}_\ccalM \Big[ \!\sum_{i=1}^n [\bbPhi(\bbx;\bbS_{P:1},\ccalH_{t+1})]^2_i \Big] \Big)^2  \\
	&- 2 \eta_\gamma (\gamma_{1}^* \!-\! \gamma_{1,t})\Big(C_f\!-\!\frac{1}{n}{\mathbb{E}}_\ccalM \Big[ \!\sum_{i=1}^n [\bbPhi(\bbx;\bbS_{P:1},\ccalH_{t+1})]_i \Big] \Big) \nonumber\\
	&- 2 \eta_\gamma (\gamma_{2, t} \!-\! \gamma_2^*)\Big(C_s\!-\!\frac{1}{n}{\mathbb{E}}_\ccalM \Big[ \!\sum_{i=1}^n [\bbPhi(\bbx;\bbS_{P:1},\ccalH_{t+1})]^2_i \Big] \Big). \nonumber
\end{align}
We now analyze the terms in \eqref{eq:proofThm4eq25} separately. 

\smallskip
\noindent \textbf{Second and third terms.} Using the triangle inequality and the condition $\| \bbPhi(\bbx;\bbS_{P:1},\ccalH_{t+1}) \| \le C_y$ [As. \ref{as:inputBound}], we have
\begin{align}\label{eq:proofThm4eq4}
	&\Big( C_f-\frac{1}{n}{\mathbb{E}}_\ccalM\! \Big[ \!\sum_{i=1}^n [\bbPhi(\bbx;\bbS_{P:1},\ccalH_{t+1})]_i \Big] \Big)^2 \\
	&\le \Big( C_f + \frac{1}{n}{\mathbb{E}}_\ccalM\! \Big[ \| \bbPhi(\bbx;\bbS_{P:1},\ccalH_{t+1}) \|_1 \Big]\Big)^2 \nonumber \le \Big( C_f + \frac{C_y}{\sqrt{n}} \Big)^2
\end{align}
where we also use $\|\cdot\|_1 \le \sqrt{n}\| \cdot \|$. Similarly we have
\begin{align}\label{eq:proofThm4eq5}
	&\Big( C_s\!-\frac{1}{n}{\mathbb{E}}_\ccalM\! \Big[ \!\sum_{i=1}^n [\bbPhi(\bbx;\bbS_{P:1},\ccalH)]^2_i \Big] \Big)^2 \!\le\! \Big( C_s \!+\! \frac{C_y^2}{n} \Big)^2.
\end{align}

\noindent \textbf{Forth and fifth terms.} To analyze these two terms, we first consider the difference of the dual problem $\ccalD(\bbgamma)$ [cf. \eqref{eq:dualFunction}] evaluated at the optimal dual variable $\bbgamma^*$ and an arbitrary one $\bbgamma$ as
\begin{align}\label{eq:proofThm4eq1}
	\ccalD(\bbgamma^*) - \ccalD(\bbgamma) & = \min_\ccalH \ccalL(\ccalH, \bbgamma^*) - \min_\ccalH \ccalL(\ccalH, \bbgamma) \\
	&\le \ccalL(\ccalH^{(\Gamma)}, \bbgamma^*) -  \ccalL(\ccalH^*, \bbgamma) \nonumber
\end{align}
where $\ccalH^*$ is the optimal solution of $\min_\ccalH \ccalL(\ccalH, \bbgamma)$ and $\ccalH^{(\Gamma)}$ is the solution of $\min_\ccalH \ccalL(\ccalH, \bbgamma)$ obtained by the gradient descent at the primal phase [cf. \eqref{eq_priup}]. From Assumption \ref{as:primalOptimalityLoss}, we have $\ccalL(\ccalH^*, \bbgamma) \ge \ccalL(\ccalH^{(\Gamma)}, \bbgamma) - \xi$. Substituting this result and the Lagrangian expression \eqref{eq:Lagrangian} into \eqref{eq:proofThm4eq1}, we get
\begin{align}\label{eq:proofThm4eq2}
	&\ccalD(\bbgamma^*) - \ccalD(\bbgamma) \le \ccalL(\ccalH^{(\Gamma)}, \bbgamma^*) -  \ccalL(\ccalH^{(\Gamma)}, \bbgamma) + \xi \\
	& = (\gamma_1^* - \gamma_1) \Big( C_f\!-\!\frac{1}{n}{\mathbb{E}}_\ccalM\! \Big[ \!\sum_{i=1}^n [\bbPhi(\bbx;\bbS_{P:1},\ccalH^{(\Gamma)})]_i \Big] \Big)\nonumber\\
	& ~ + (\gamma_2 - \gamma_2^*) \Big( C_s\!-\!\frac{1}{n}{\mathbb{E}}_\ccalM\! \Big[ \!\sum_{i=1}^n [\bbPhi(\bbx;\bbS_{P:1},\ccalH^{(\Gamma)})]^2_i \Big] \Big) + \xi.\nonumber
\end{align}

By using the fact $\ccalH_{t+1} = \ccalH^{(\Gamma)}_t$, we substitute \eqref{eq:proofThm4eq4}, \eqref{eq:proofThm4eq5} and \eqref{eq:proofThm4eq2} into \eqref{eq:proofThm4eq25} and altogether into \eqref{eq:proofThm4eq225} to get
\begin{align}\label{eq:proofThm4eq6}
	\| \bbgamma_{t+1} \!-\! \bbgamma^* \!\|^2 &\le \| \bbgamma_{t} \!-\! \bbgamma^* \|^2 \!+\! \eta_\gamma A_t
\end{align}
with $A_t = \eta_\gamma ((C_f + C_y/\sqrt{n})^2 + (C_s + C_y^2/n)^2 ) + 2 (\ccalD(\bbgamma_t) - \ccalD(\bbgamma^*) + \xi)$. This expression characterizes the update progress of the dual step. By unrolling \eqref{eq:proofThm4eq6} to the initialization $t=0$, we get 
\begin{align}\label{eq:proofThm4eq7}
	\| \bbgamma_{t\!+\!1} \!-\! \bbgamma^* \|^2 \!&\le\! \| \bbgamma_{t} \!-\! \bbgamma^* \|^2 \!+\! \eta_\gamma A_t \!\le\! \| \bbgamma_0 \!-\! \bbgamma^* \|^2 \!+\! \sum_{i=0}^{t} \eta_\gamma A_i.
\end{align}
Since $\ccalD(\bbgamma^*)$ is the maxima of $\ccalD(\bbgamma)$, $\ccalD(\bbgamma_t) - \ccalD(\bbgamma^*)$ is always negative. Therefore, when $\bbgamma_t$ is far from $\bbgamma^*$, the difference $\ccalD(\bbgamma_t) - \ccalD(\bbgamma^*)$ is largely negative, and thus $A_t$ is also negative. Consider the iteration number $T = \argmin_{t \in \mathbb{Z}} A_t > -2 \delta$ where $\delta$ is the desirable accuracy. Then, we have $A_t \le -2 \delta$ for all $0 \le t < T$ and substituting this result into \eqref{eq:proofThm4eq7} yields 
\begin{align}\label{eq:proofThm4eq8}
	\| \bbgamma_{T} - \bbgamma^* \|^2 \le \| \bbgamma_0 - \bbgamma^* \|^2 - 2 T \eta_\gamma \delta.
\end{align}
Since $\| \bbgamma_{T} - \bbgamma^* \|^2 \ge 0$, we get $T \le \|\bbgamma_0 - \bbgamma^* \|^2/(2\eta_\gamma\delta)$ which indicates that $T$ is finite and bounded. By substituting the expression of $A_T$ into the condition $A_T > -2 \delta$, we obtain
\begin{align}
	D(\!\bbgamma_T\!) \!\le\! \ccalD(\!\bbgamma^*\!) \!<\!\! D(\!\bbgamma_T\!) \!+\! \frac{\eta_\gamma \Big( \!\!\big(\!C_f \!+\! \frac{C_y}{\sqrt{n}}\big)^2 \!\!\!+\! \big(C_s \!+\! \frac{C_y^2}{n}\big)^2 \!\Big)}{2} \!+\! \xi \!+\! \delta. \nonumber
\end{align}
Further leveraging the fact $|\ccalL(\ccalH^{(\Gamma)}_T, \bbgamma_T) - \ccalD(\bbgamma_T)| \le \xi$, we get
\begin{align}
	|\ccalL(\ccalH^{(\Gamma)}_T\!,\! \bbgamma_T)\!-\! D(\bbgamma^*)| \!\le\!\!  \frac{\eta_\gamma \Big(\! \!\big(\!C_f \!+\! \frac{C_y}{\sqrt{n}}\big)^2 \!\!\!+\! \big(\!C_s \!+\! \frac{C_y^2}{n}\!\big)^2 \Big)}{2} \!+\! 2\xi \!+\! \delta \nonumber
\end{align}
which completes the proof.

\section{Lemmas and Proofs}\label{lemmasProof}

\begin{lemma} \label{lemma:errorMatrixExpectation}
	Consider the nominal graph $\ccalG$ with the shift operator $\bbS$ and the GRES($p,q$) model [Def. \ref{def_res}]. Let $\ccalG_d$ be the subgraph representing $M_d$ existing edges that may be dropped with the degree matrix $\bbD_d$ and the shift operator $\bbS_d$, and $\ccalG_a$ be the subgraph representing $M_a$ new edges that may be added with the degree matrix $\bbD_a$ and the shift operator $\bbS_a$. Let also $\bbS_k$ be the shift operator of the $k$th GRES($p,q$) graph realization and $\barbS=\mathbb{E}[\bbS_k]$ the expected shift operator. Then, it holds that
	\begin{equation}\label{lemma2:main}
		\mathbb{E}\left[ \bbS_k^2 \right] \!\!=\! 
		\begin{cases}
			\barbS^2 \!+\! p (1\!-\!p)\bbD_d \!+\! q(1\!-\!q)\bbD_a,\! & \!\text{if } \bbS \!=\! \bbA, \\
			\barbS^2 \!+\! 2 p (1\!-\!p)\bbS_d \!+\! 2q(1\!-\!q)\bbS_a,\! & \!\text{if } \bbS \!=\! \bbL
		\end{cases}
	\end{equation}
	where $\bbA$ is the adjacency matrix and $\bbL$ is the Laplacian matrix of the nominal graph $\ccalG$. 
\end{lemma}

\begin{proof}
	The GRES($p,q$) model drops edges in $\ccalG_d$ with probability $p$ and adds edges in $\ccalG_a$ with probability $q$ independently. Let $\bbS_k = \bbS_{d,k} + \bbS_{a,k}$ be a GRES($p,q$) realization of $\bbS$ where $\bbS_{d,k}$ and $\bbS_{a,k}$ are realizations of $\bbS_d$ and $\bbS_a$, and $\barbS = \barbS_d + \barbS_a$ be the expected shift operator of $\bbS$ where $\barbS_d$ and $\barbS_a$ are the expected shift operators of $\bbS_d$ and $\bbS_a$, respectively. By substituting the latter into $\bbS_k^2$, we have
	\begin{align}\label{proof:lemma2eq1}
		\mathbb{E}\!\left[ \bbS_k^2 \right]\!\!=\! \mathbb{E}\!\left[(\bbS_{d,k} \!+\! \bbS_{a,k}\!)^2 \right] \!\!=\! \mathbb{E}\!\left[\bbS_{d,k}^2\!\! +\! 2 \bbS_{d,k} \bbS_{a,k} \!\!+\! \bbS_{a,k}^2 \right]\!\!.
	\end{align}
	
	\noindent \textbf{Adjacency.} For the adjacency matrix $\bbS = \bbA$, we have from Lemma 2 in \cite{gao2021stochastic} that
	\begin{subequations} \label{proof:lemma2eq15}
		\begin{align}\label{proof:lemma2eq2}
			\mathbb{E}\left[ \bbS_{d,k}^2 \right] \!=\! \barbS_d^2 + p (1\!-\!p)\bbD_d, \\
			\label{proof:lemma2eq3} \mathbb{E}\left[ \bbS_{a,k}^2 \right] \!=\! \barbS_a^2 + q (1\!-\!q)\bbD_a. 
		\end{align}
	\end{subequations} 
	Since $\bbS_{d,k}$ and $\bbS_{a,k}$ are mutually independent, we have
	\begin{align}\label{proof:lemma2eq4}
		\mathbb{E}\left[ \bbS_{d,k}\bbS_{a,k} \right] \!=\! \barbS_d \barbS_a.
	\end{align}
	By substituting \eqref{proof:lemma2eq15} and \eqref{proof:lemma2eq4} into \eqref{proof:lemma2eq1}, we get
	\begin{align}\label{proof:lemma2eq5}
		\mathbb{E}\!\left[ \bbS_k^2 \right]\!\!&=\!\barbS_d^2 \!+\!\barbS_a^2 \!+\! p (1\!-\!p)\bbD_d \!+\!q (1\!-\!q)\bbD_a + 2 \barbS_d \barbS_a\\
		& =\! ( \barbS_d +  \barbS_a)^2\!+\! p (1\!-\!p)\bbD_d \!+\!q (1\!-\!q)\bbD_a. \nonumber 
	\end{align}
	By leveraging the fact $\barbS^2 = (\barbS_d + \barbS_a)^2$ in \eqref{proof:lemma2eq5}, we have
	\begin{align}\label{proof:lemma2eq6}
		\mathbb{E}\!\left[ \bbS_k^2 \right]\!=\!\barbS^2 + p (1\!-\!p)\bbD_d + q(1\!-\!q)\bbD_a.
	\end{align}
	
	\noindent \textbf{Laplacian.} For the Laplacian matrix $\bbS = \bbL$, we have again from Lemma 2 in \cite{gao2021stochastic} that 
	\begin{subequations}
		\begin{align}\label{proof:lemma2eq7}
			\mathbb{E}\left[ \bbS_{d,k}^2 \right] \!=\! \barbS_d^2 + 2p (1\!-\!p)\bbS_d, \\
			\label{proof:lemma2eq8} \mathbb{E}\left[ \bbS_{a,k}^2 \right] \!=\! \barbS_a^2 + 2q (1\!-\!q)\bbS_a. 
		\end{align}
	\end{subequations}
	Following the same process as \eqref{proof:lemma2eq4}-\eqref{proof:lemma2eq6}, we get
	\begin{align}\label{proof:lemma2eq9}
		\mathbb{E}\!\left[ \bbS_k^2 \right]\!=\!\barbS^2 + 2p (1\!-\!p)\bbS_d + 2q(1\!-\!q)\bbS_a
	\end{align}
	which completes the proof.
\end{proof}

\smallskip

\begin{lemma}\label{lemma:outputBound}
	Consider the SGNN $\bbPhi(\bbx;\bbS_{P:1}, \ccalH)$ of $L$ layers, each comprising $F$ filters with the frequency response \eqref{eq:generalizeFrequencyResponse} satisfying Assumption \ref{as:LipschitzFilter} and the nonlinearity $\sigma(\cdot)$ satisfying Assumption \ref{as:LipschitzNonlinearity1} with $C_\sigma$. Then, for any input signal $\bbx$ with a finite energy $\|\bbx\| < \infty$, there exists a constant $C_y$ such that
	\begin{align}
		\|\bbPhi(\bbx;\!\bbS_{P:1},\! \ccalH)\| \le C_y.
	\end{align}
\end{lemma}

\begin{proof}
	We start by considering the stochastic graph filter $\bbH(\bbS_{K:0})$. By conducting a chain of GFTs on the input signal $\bbx$ [cf. \eqref{eq:GFTChain}], we have
	\begin{align}\label{proof:lemma1eq1}
		&\bbH(\bbS_{K:0}) \bbx \\
		&=\!\! \sum_{i_0=1}^N\sum_{i_1=1}^N \!\!\cdots\!\!\! \sum_{i_K=1}^N \!\!\!\hat{x}_{0i_0}\hat{x}_{1i_{0}i_1}\cdots \hat{x}_{Ki_{K\!-\!1}i_K} \!\!\sum_{k=0}^K \!h_k\! \prod_{j=0}^k\!\lambda_{ji_j} \bbv_{Ki_K}.\nonumber
	\end{align} 
	Since from Assumption \ref{as:LipschitzFilter} it holds that $|h(\bblambda)| \le 1$, we have $|\sum_{k=0}^K h_k \prod_{j=0}^k\lambda_{ji_j}| \le 1$. By substituting this result and the orthogonality of eigenvectors $\{\bbv_{Ki_K}\}_{i_K=1}^n$ into \eqref{proof:lemma1eq1}, we have.
	\begin{align}\label{proof:lemma1eq2}
		\|\bbH(\bbS_{K:0}) \bbx \| \le \|\bbx\|.
	\end{align} 
	
	We then consider the SGNN output, whose norm can be bounded by
	\begin{align}\label{proof:lemma1eq3}
		&\|\bbPhi(\bbx;\!\bbS_{P:1},\! \ccalH)\| \!:=\! \|\bbx_{L}^{1}\| \!=\! \Big\|\sigma\Big(\sum_{g=1}^{F} \bbH_{L}^{1g}(\bbS_{K:0})\bbx_{L-1}^g \Big)\Big\|\nonumber\\
		& \le C_\sigma \sum_{g=1}^F \big\| \bbH_{\ell}^{1g}(\bbS_{K:0})\bbx_{L-1}^g \big\| \le C_\sigma \sum_{g=1}^F \| \bbx_{L-1}^g \|
	\end{align} 
	where in the first inequality we used the Lipschitz condition of the nonlinearity $\sigma(\cdot)$ from Assumption \ref{as:LipschitzNonlinearity1} and the triangle inequality, and in the second inequality we used \eqref{proof:lemma1eq2}. Unrolling \eqref{proof:lemma1eq3} recursively until the input layer and proceeding in the same way, we have
	\begin{align}\label{proof:lemma1eq4}
		&\|\bbPhi(\bbx;\!\bbS_{P:1},\! \ccalH)\| \le C_\sigma^L F^{L-1} \| \bbx \|
	\end{align} 
	where $C_y = C_\sigma^L F^{L-1} \| \bbx \|$ is a finite constant. The latter completes the proof.
\end{proof}


\smallskip

\begin{lemma} \label{lemma:stabilitySGNN}
	Consider the SGNN $\bbPhi(\bbx;\bbS_{P:1},\ccalH)$ of $L$ layers comprising $F$ stochastic graph filters of order $K$ [cf. \eqref{eq:sgnn}] and the frequency response \eqref{eq:generalizeFrequencyResponse} satisfying Assumption \ref{as:LipschitzFilter} with $C_L$. Let the nonlinearity $\sigma(\cdot)$ satisfy Assumption \ref{as:LipschitzNonlinearity1} with $C_\sigma$ and $\bbS_{P:1} = \{\bbS_1,\ldots,\bbS_P\}$, $\widetilde{\bbS}_{P:1} = \{\widetilde{\bbS}_1,\ldots,\widetilde{\bbS}_P\}$ be two sequences of random shift operators satisfying $\|\widetilde{\bbS}_k - \bbS_k \| \le \varepsilon$ for all $k=1,\ldots,P$. Then, for any input signal $\bbx$ with a finite energy $\|\bbx\| < \infty$, it holds that
	\begin{align}\label{lemma3:main}
		\| \bbPhi(\bbx;\widetilde{\bbS}_{P:1},\ccalH) - \bbPhi(\bbx;\bbS_{P:1},\ccalH) \| \le C_B \varepsilon + \ccalO(\varepsilon^2)
	\end{align}
	where $C_B = K C_L L C_\sigma^L F^{L-1} \|\bbx\|$ is a constant.
\end{lemma}

\begin{proof}
	We start by considering the stochastic graph filter $\bbH(\bbS_{K:0})$. From $\|\widetilde{\bbS}_k - \bbS_k \| \le \varepsilon$ for all $k=1,\ldots,K$, write $\widetilde{\bbS}_k = \bbS_k + \bbE_k$ where $\bbE_k$ is the $k$th error matrix that characterizes the deviation of $\widetilde{\bbS}_k$ from $\bbS_k$. By substituting this representation into the filter and expanding the terms, we have
	\begin{align}\label{proof:lemma3eq1}
		\bbH(\widetilde{\bbS}_{K:0}) \!-\! \bbH(\bbS_{K:0}) \!=\! \sum_{k=1}^K \sum_{r = k}^K \!h_r\!\!\!\! \prod_{\tau=k\!+\!1}^r \!\!\bbS_\tau \bbE_k\! \prod_{j=1}^{k\!-\!1}\!\bbS_{j} \!+\! \bbC
	\end{align}
	where we define $\prod_{a}^b (\cdot) = 0$ if $a > b$ and $\bbC$ collects the rest expanding terms that contains at least two error matrices $\bbE_{k_1}$ and $\bbE_{k_2}$ with $k_1 \ne k_2$. Since $\|\bbE_k\| \le \varepsilon$ for all $k=1,\ldots,K$ and the filter coefficients $\{h_k\}_{k=1}^K$ are finite, we have $\| \bbC \| = \ccalO(\varepsilon^2)$. Let us then consider the $K$ terms $\{\sum_{r = k}^K h_r \prod_{\tau=k+1}^r \!\bbS_\tau \bbE_k \prod_{j=1}^{k-1}\!\bbS_{j}\}_{k=1}^K$ separately.
	
	For each term, let $\bbx$ be the input signal and $\bbE_k = \bbU_k \bbM_k \bbU_k$ be the eigendecomposition with eigenvectors $\bbU_k = [\bbu_{k1},\ldots,\bbu_{kn}]$ and eigenvalues $\bbM_k = \diag(m_{k1},\ldots,m_{kn})$. By conducting a chain of GFTs on the input signal $\bbx$ over $\{\bbS_1,\ldots,\bbS_{k-1},\bbE_k,\bbS_{k+1},\ldots,\bbS_r\}$[cf. \eqref{eq:GFTChain}], we have
	\begin{align}\label{proof:lemma3eq2}
		&\sum_{r = k}^K h_r\! \prod_{\tau=k\!+\!1}^r \!\bbS_\tau \bbE_k\! \prod_{j=1}^{k\!-\!1}\!\bbS_{j} \bbx \\
		&=\!\!\! \sum_{i_1=1}^N \!\!\ldots\!\!\! \sum_{i_K=1}^N \!\!\!\hat{x}_{1i_1}\ldots \hat{x}_{Ki_{K\!-\!1}i_K}\!\! \sum_{r=k}^K\! h_r\!\!\nonumber \prod_{\tau\!=\!k\!+\!1}^r\!\!\!\lambda_{\tau i_\tau} m_{ki_k} \prod_{j=1}^{k\!-\!1}\!\lambda_{j i_j} \bbv_{Ki_K}\!.\nonumber
	\end{align} 
	Consider the term $\sum_{r=k}^K h_r \prod_{\tau=k+1}^r\lambda_{\tau i_\tau} m_{ki_k} \prod_{j=1}^{k\!-\!1}\!\lambda_{j i_j}$ in \eqref{proof:lemma3eq2}. Let $\bblambda_{\bbE_k} \!\!=\! [\lambda_{1i_1}\!,...,\!\lambda_{(k\!-\!1)i_{k\!-\!1}}\!,m_{ki_k}\!,\lambda_{(k\!+\!1)i_{k\!+\!1}}\!,...,\!\lambda_{Ki_K}]^\top$ be an instantiation of the $K$-dimensional vector variable $\bblambda=[\lambda_{1},\ldots,\lambda_{K}]^\top$ in the frequency response function $h(\bblambda)$ [cf. \eqref{eq:generalizeFrequencyResponse}]\footnote{The multivariate frequency response function $h(\bblambda)$ is an analytic function of the vector variable $\bblambda=[\lambda_{1},\ldots,\lambda_{K}]^\top$ such that $\bblambda$ can take any value.}. The partial derivative of $h(\bblambda)$ over the $k$th variable entry $\lambda_k$ evaluated at the instantiation $\bblambda_{\bbE_k}$ is given by
	\begin{align}\label{proof:lemma3eq4}
		&\frac{\partial h (\bblambda_{\bbE_k})}{\partial \lambda_k} \!= \! \sum_{r=k}^K h_r\!\! \prod_{\tau=k+1}^r\!\!\lambda_{\tau i_\tau} \prod_{j=1}^{k\!-\!1}\!\lambda_{j i_j}.
	\end{align}
	From the Lipschitz property of the frequency response in Assumption \ref{as:LipschitzFilter}, we have
	\begin{align}\label{proof:lemma3eq45}
		\Big| \frac{\partial h (\bblambda_{\bbE_k})}{\partial \lambda_k} \Big| \!= \! \Big|\sum_{r=k}^K h_r\!\!\!\! \prod_{\tau=k\!+\!1}^r\!\!\!\lambda_{\tau i_\tau}\! \prod_{j=1}^{k\!-\!1}\!\lambda_{j i_j}\!\Big| \!\le\! C_L
	\end{align}	
	and from the fact $\| \bbE_k \| \le \varepsilon$, we have
	\begin{align}\label{proof:lemma3eq5}
		|m_{k i_k}| \!\le\! \varepsilon,~\forall~k=1,\ldots,K~\text{and}~i_k=1,\ldots,n.
	\end{align}
	By leveraging \eqref{proof:lemma3eq45} and \eqref{proof:lemma3eq5}, we get
	\begin{align}\label{proof:lemma3eq6}
		&\Big|\sum_{r=k}^K h_r\! \prod_{\tau=k\!+\!1}^r\!\lambda_{\tau i_\tau} m_{ki_k} \prod_{j=1}^{k\!-\!1}\!\lambda_{j i_j}\Big| \\
		& \le \Big|\sum_{r=k}^K h_r\! \prod_{\tau=k\!+\!1}^r\!\lambda_{\tau i_\tau} \prod_{j=1}^{k\!-\!1}\!\lambda_{j i_j}\Big| \Big|m_{k i_k}\Big| \le C_L \varepsilon. \nonumber
	\end{align}
	
	By using the orthogonality of eigenvectors $\{\bbv_{Ki_K}\}_{i_K=1}^n$ and \eqref{proof:lemma3eq6} in \eqref{proof:lemma3eq2}, we get
	\begin{align}\label{proof:lemma3eq7}
		&\Big\|\sum_{r = k}^K h_r\! \prod_{\tau=k\!+\!1}^r \!\bbS_\tau \bbE_k\! \prod_{j=1}^{k\!-\!1}\!\bbS_{j} \bbx\Big\|^2 \\
		& = \sum_{i_1=1}^N \!\!\ldots\!\!\! \sum_{i_K\!=\!1}^N \!\Big|\hat{x}_{1i_1}\ldots \hat{x}_{Ki_{K\!-\!1}i_K} \sum_{r=k}^K h_r\!\!\! \prod_{\tau=k\!+\!1}^r\!\!\!\lambda_{\tau i_\tau} m_{ki_k} \prod_{j=1}^{k\!-\!1}\!\lambda_{j i_j}\Big|^2 \nonumber\\
		&\le \sum_{i_1=1}^N \!\!\ldots\!\!\! \sum_{i_K=1}^N \!\!\Big| \hat{x}_{1i_1}\ldots \hat{x}_{Ki_{K\!-\!1}i_K} \Big|^2 (C_L \varepsilon)^2 = C_L^2 \|\bbx\|^2 \varepsilon^2.\nonumber
	\end{align} 
	By further using the triangle inequality and \eqref{proof:lemma3eq7} in \eqref{proof:lemma3eq1}, we get
	\begin{align}\label{proof:lemma3eq8}
		&\|\bbH(\widetilde{\bbS}_{K:0}) \bbx\!-\! \bbH(\bbS_{K:0})\bbx\| \\
		&\le\! \sum_{k=1}^K\! \big\| \sum_{r = k}^K \!h_r\!\!\!\! \prod_{\tau=k\!+\!1}^r \!\!\bbS_\tau \bbE_k\! \prod_{j=1}^{k\!-\!1}\!\bbS_{j}\big\|\!+\! \|\bbC\| \!\le\! K C_L \| \bbx \| \varepsilon \!+\! \ccalO(\varepsilon^2).\nonumber
	\end{align}
	By combining \eqref{proof:lemma3eq8} with steps Eq. (79)-(89) in the proof of Theorem 4 in \cite{Gama20-Stability}, we complete the proof
	\begin{align}\label{proof:lemma3eq9}
		&\| \bbPhi(\bbx;\widetilde{\bbS}_{P:1},\ccalH) - \bbPhi(\bbx;\bbS_{P:1},\ccalH) \| \le C_B \varepsilon + \ccalO(\varepsilon^2) 
	\end{align}
	where $C_B = K C_L L C_\sigma^L F^{L-1} \|\bbx\|$ is a finite constant. The latter completes the proof.
\end{proof}


\bibliographystyle{IEEEbib}
\bibliography{myIEEEabrv,biblioOp}

\end{document}